\documentclass[11pt, reqno, a4paper]{article}

\usepackage{amsmath,amssymb,amsthm}
\usepackage[margin=1in]{geometry}
\usepackage{graphicx}
\usepackage{subcaption}
\usepackage{adjustbox}

\usepackage{tikz}
\usepackage{hyperref}
\usepackage{xcolor}
\usepackage{titlesec}
\usepackage{bbm}
\usepackage{tocloft}
\usepackage{enumitem}
\usepackage{comment}
\usepackage{etoolbox}
\usepackage{cleveref}

\allowdisplaybreaks
\newtheorem{theorem}{Theorem}[section]
\newtheorem{lemma}[theorem]{Lemma}

\theoremstyle{definition}
\newtheorem{definition}[theorem]{Definition}
\newtheorem{example}[theorem]{Example}
\newtheorem{prop}[theorem]{Proposition}

\tikzset{
  font={\fontsize{5pt}{12}\selectfont}}
\theoremstyle{remark}
\newtheorem{remark}[theorem]{Remark}

\numberwithin{equation}{section}

\title{On Trajectory-Based Stability Analysis for $1$-bit Sigma-Delta Quantization and its Application to the Second-Order Case}

\author{Rohan Joy, Felix Krahmer, Alessandro Lupoli}

\date{} % lascia vuoto per non mostrare la data
\begin{document}
\maketitle

\begin{abstract}
\noindent
A state-of-the-art strategy for digitally representing a bandlimited signal $f$ is Sigma-Delta $(\Sigma\Delta)$ quantization. $\Sigma\Delta$ quantization schemes choose a bit sequence $(q_n)$ representing the samples $(y_n)$ of $f$ sequentially based on a state sequence $(u_n)$ defined via a recurrence relation of the form 
\begin{equation*}
    u_n = (h*u)_n + y_n - q_n,
\end{equation*}
where $h$ is a causal filter, i.e., $h_j = 0$ for $j\le 0.$ The effectiveness of a quantization scheme crucially depends on the fact that it is stable, that is, the state variable remains uniformly bounded in a given class of signals. Thus, a natural (and common) strategy to choose the bit representation is to minimize the state variable, i.e., $$q_n = \operatorname{sign}((h*u)_n + y_n).$$
It is well known that a sufficient condition for this quantization rule to induce stability is that $$ \|h\|_{\ell^1}+\|f\|_{\infty}\le 2.$$
At the same time, one empirically observes that this condition is conservative and stability holds significantly beyond this bound.

In this paper, we address this gap by establishing the first stability guarantees beyond first order that outperform the $\ell^1$ based stability condition.
In contrast to many previous approaches, our analysis describes the trajectories of the state variables rather than characterizing the invariant set, an approach that had previously been performed only in some specific example cases. The advantage of this viewpoint is that it allows for longer filters, which is challenging when analyzing invariant sets, due to the resulting large dimension, and perturbations resulting from non-constant inputs can be better incorporated. 

Concretely, we apply our analysis technique to second-order $\Sigma\Delta$ schemes with sparse feedback filters as proposed by G\"unturk \cite{gunturk2003one}, showing that the filter length required to guarantee stability significantly improves from the length $O\left(\frac{1}{1-\|f\|_{\infty}}\right)$ needed to apply the $\ell^1$ based criterion to $O\left(\frac{1}{\sqrt{1-\|f\|_{\infty}}}\right)$.

\end{abstract}

\section{Introduction}
Digital representations of analog signals are of fundamental importance in modern signal processing. More precisely, one aims to encode a continuous-time signal or function $f$ as a sequence with values in a finite set, called alphabet $\mathcal{A},$ that then allows for approximate reconstruction of $f$ at least in a domain of interest.
The process of analog-to-digital (A/D) conversion typically proceeds by first \emph{sampling} the function, that is, evaluating it at discrete time instances $(t_n)_{n \in \mathbb{N}}$. The sequence of samples is then \emph{quantized}, i.e., represented by a sequence $(q_n)_{n\in\mathbb{N}}$ with $q_n\in \mathcal{A}.$
Consequently, the key design feature of an A/D conversion method is to find an appropriate \emph{quantization scheme} 
that guarantees a small reconstruction error in this pipeline.\\

In practice, most A/D architectures are based on the Nyquist–Shannon sampling theorem, 
which guarantees perfect reconstruction of bandlimited functions from uniformly spaced samples 
when the sampling rate is at least the Nyquist rate given by the bandwidth $\Omega$ -- twice the maximum frequency present in the signal.\\
This motivates working with equispaced samples and considering a reconstruction method where the Nyquist-Shannon formula is applied to the sequence of quantized values.
A natural and straightforward quantization technique uses the so–called 
\emph{memoryless scalar quantizer} (MSQ), which simply rounds each sample to the closest element 
of the alphabet $\mathcal{A}$.  
However, it is quite intuitive that a high–resolution reconstruction from these quantized values can only be obtained 
when the alphabet is resolving the dynamic range in a sufficiently fine way.\\
This is often not desired, as enlarging the alphabet -- or equivalently, increasing the number of quantization bits per sample -- increases the \emph{circuit complexity} of the converter, which is linked to a higher power consumption.  
While keeping the number of bits per sample fixed and increasing the sampling rate also results in a larger total number of bits, and hence, in a higher power consumption, such an approach is nevertheless preferred  in many practical implementations,
as the power consumption of A/D converters scales approximately 
\emph{linearly with the sampling frequency} (due to faster analog front-end operation) 
and \emph{exponentially with the number of bits} in the quantizer, 
as higher resolution requires more precise and power-demanding analog components.\\
For this reason, it is often preferable to \emph{oversample} the input signal, 
acquiring samples at a rate significantly higher than the Nyquist frequency, 
and to exploit this \emph{redundancy} to achieve higher effective resolution 
while keeping the alphabet (and thus the number of bits per sample) small.
Making use of this redundancy requires choosing the quantized values in an interdependent fashion such that the quantization error incurred in a neighborhood of samples locally cancels.\\

The pioneering approach that first put this paradigm into practice was the $\Sigma\Delta$ quantization scheme, developed by H.~Inose and Y.~Yasuda as early as 1962 \cite{inose1962telemetering}.
The idea is that the quantized representation is driven by an
\emph{internal state variable} that is computed via a recursive scheme and can be interpreted as a higher-order variant of the cumulative quantization error.
Each quantized sample is then chosen in a greedy way, minimizing the magnitude of the corresponding iterate of the state variable.\\
Over the years, $\Sigma\Delta$ converters have become the dominant architecture for high-resolution, low-to-medium-bandwidth applications such as audio and precision measurements, see \cite{hoeschele1994analog,candy1991oversampling,schreier2005understanding}.

Key to the success of $\Sigma\Delta$ quantization is that the scheme is stable, that is, the state variable remains bounded. If this is not the case, one will not only incur hardware saturation, but the overall objective to keep the cumulative quantization error under control will also be violated.
Consequently, stability guarantees are also a key ingredient for any theoretical performance analysis.
Nevertheless, all results establishing stability guarantees are limited in some way.\\
For the simplest class of $\Sigma\Delta$ quantizers that use the standard cumulative quantization error as a state variable, so-called first-order schemes, stability is straightforward and does not require any restriction of the signal amplitude.
However, the error decay is limited, so high oversampling is required to achieve a good reconstruction accuracy \cite{gunturk2001robustness}.\\
To establish stability for schemes with better error decay, a pioneering paper considers constant inputs and a specific choice of coefficients in the recursive scheme \cite{hein2002stability}, and provides an analysis of the trajectory of the internal state variable. However, it is unclear how this analysis can be carried over to more general time-varying signals.
% In fact, it has been shown in a different context that $\Sigma\Delta$ schemes can behave quite differently for time-varying signals as compared to constant signals [].\\
Another line of work investigates $\Sigma\Delta$ schemes with modified quantization rules. The most general analysis covers a wide class of schemes, allowing for fast error decay, but involves nested constructions of quantization operations that are less preferred from an application point of view due to their circuit complexity \cite{daubechies2003approximating}. A second paper focusing on a more restrictive class of schemes is simpler in that respect, but parameter configurations that allow for large signal amplitudes can lead to very large state variables, even when the actual signal's amplitude is far from the maximum admissible bound \cite{yilmaz2002stability}.\\
A third line of work shows that even for schemes with very fast error decays, one can achieve a state variable bounded by $1$ at the expense of increased memory usage in the recurrence relation.\\
In particular, the number of values of the state variables that need to be stored grows rapidly with the admissible amplitude.
For this reason, practitioners have suggested mildly violating the stability criterion underlying these guarantees, as it has been empirically observed to be conservative (see \cite{schreier2005understanding}). 
Nevertheless, we are not aware of any theoretical analysis supporting this claim.
In this paper, we are making the first step towards closing this gap by rigorously establishing that for the class of second-order $\Sigma\Delta$ schemes proposed in \cite{deift2011optimal,gunturk2003one}, stability can be achieved for bandlimited signals by storing considerably fewer state variables. We show that the state variables only get large as the signal amplitude approaches the maximum admissible value and returns to smaller values once the signal amplitude decreases again.\\\\
The paper is organized as follows: Section 2 and $3$ introduces the mathematical background of $\Sigma\Delta$ quantization, formulate the stability problem for second-order schemes, and discuss our main results. Section $4$ begins with an analysis of the trajectories of the internal state variable of a $\Sigma\Delta$ scheme—a crucial step for establishing stability beyond the guarantees of the classical criterion. We address this first for constant inputs and then extend the analysis to bandlimited signals. 
% The final section discusses refinements to the analysis in Section 2 and presents a specific case of a $\Sigma\Delta$ scheme.

\section{Background and Main Result}
In this section, we will formulate the aforementioned findings in a more rigorous way, and discuss their relation to the results of this paper. We start by precisely introducing $1$-bit $\Sigma\Delta$ quantization.
\subsection{Sigma-Delta quantization}
In this paper, we consider bandlimited signals in the Paley-Wiener space
\begin{equation}
    PW_{\Omega}(\mathbb{R}):= \left\{f\in L^2(\mathbb{R}) \ | \ \text{supp}(\hat{f})\subset \left[-\frac{\Omega}{2},\frac{\Omega}{2}\right] \right\},
\end{equation} 
where we use the following normalization for the Fourier transform of a signal $f$
\begin{equation}
    \hat{f}(\xi) = \int_{\mathbb{R}}f(t)e^{-2\pi i \xi t}dt \quad \text{for all} \ \xi\in \mathbb{R}.
\end{equation}
For notational convenience, we normalize to $\Omega = 1.$
We assume oversampling with respect to the Nyquist rate, that is, a sampling rate satisfying $\lambda\ge \lambda_0,$ for some $\lambda_0>1,$ as this allows for a version of the sampling theorem with better convergence properties.\\
Namely, one has that for such $\lambda$ and any $f\in PW_{1}(\mathbb{R})$
\begin{equation}\label{Shannon_Theorem}
    f(t) = \frac{1}{\lambda} \sum_{n \in \mathbb{Z}} f\!\left(\frac{n}{\lambda}\right) 
    \varphi\!\left(t- \frac{n}{\lambda}\right) \quad \text{for all} \quad t\in\mathbb{R}
\end{equation}
where $\varphi$ is a Schwartz function (fixed for the remainder of this paper) satisfying
\[
\hat{\varphi}(\xi) :=
\begin{cases}
    1, & \text{if } |\xi| \le \frac{1}{2}, \\
    0, & \text{if } |\xi| \ge \frac{\lambda_0}{2}.
\end{cases}
\]   
Given the representation formula \eqref{Shannon_Theorem}, one aims to replace the samples $y_n := f\left(\frac{n}{\lambda}\right)$ with values $q_n\in\{-1,1\},$ such that 
\begin{equation*}
    f_q(t) := \frac{1}{\lambda} \sum_{n \in \mathbb{Z}} q_n 
    \varphi\!\left(t- \frac{n}{\lambda}\right)
\end{equation*}
is a good approximation of $f(t)$ for all $t\in \mathbb{R}$. It is not hard to show that such a function will always have a maximum value of at most $1,$ which is why one cannot represent signals of larger amplitude in this way. For this reason, we must impose as a first constraint on the  the signals to be quantized that its amplitude is bounded by $1$.

By \eqref{Shannon_Theorem}, the reconstruction error takes the form 
\[
e(t):= f(t)-f_q(t)=  \frac{1}{\lambda} \sum_{n \in \mathbb{Z}} (y_n-q_n) 
    \varphi\!\left(t- \frac{n}{\lambda}\right)
\]
which can be interpreted as a low-pass filter applied to the \textit{quantization error} $y-q.$
In this paper, among the various metrics that can be used to quantify the reconstruction quality, we assess the error in the uniform norm and therefore focus on $\|f_q - f\|_{\infty}$.\\

To ensure that this error is small,  $\Sigma\Delta$ quantization employs a so-called noise shaping strategy. That is, one aims to choose the bit sequence $q$ in such a way that the quantization error has most of its energy in the high frequency regime, in that way ensuring that after low-pass filtering, one obtains a small reconstruction error. One way to achieve this is to find a bit sequence $(q_n)_n$ for which the frequency representation of the quantization error has multiple zeros at the origin, which can be argued to correspond to a higher-order finite difference operation of a bounded sequence $u$ of state variables in the time domain. Making this idea precise requires the theory of distributions; we refer the reader to \cite{krahmer2009novel} for a more detailed exposition. 

In line with this intuition, an $r$-th order $\Sigma\Delta$ quantization scheme aims to choose the bits $q_n$ in such a way that the equation
\begin{equation}\label{standardr}
    \Delta^r u_n = y_n - q_n,
\end{equation}
where $\Delta^r$ denotes the composition of $r$ finite difference operators, is satisfied by a bounded sequence $u$ of state variables. Here, $\Delta$ is the backwards finite difference operation, defined by 
\[
(\Delta u)_k:= u_k -u_{k-1}.
\]
When the sequence $u$ is indexed by the non-negative integers $\mathbb{N}_0:=   \{0,1,\dots\}$, we set $(\Delta u)_0=u_0$, implicitly assuming $0$ values for negative indices.

We hence seek to establish \emph{stability}, in the sense of the following definition.

\begin{definition}
Let $\mathcal{Y} \subset PW_\Omega(\mathbb{R})$ be a function space and consider a $\Sigma\Delta$ scheme of the form~\eqref{generalized_SD}.  
We say that the scheme is \emph{stable for $\mathcal{Y}$} if there exists a constant $C > 0$ such that, for every $f \in \mathcal{Y}$, the corresponding state sequence $u$ satisfies
\[
    \|u\|_\infty \le C.
\]
\end{definition}
In line with the aforementioned noise-shaping intuition, stability is directly linked to reconstruction accuracy. Indeed, one can show (see \cite{daubechies2003approximating}) that a stable $r$-th order $\Sigma\Delta$ scheme satisfies
\begin{equation}\label{error_estimate}
    \|f - f_q\|_{\infty}
    \;\le\;
    \frac{1}{\lambda^r}\,\|\varphi^{(r)}\|_{L^1}\,\|u\|_{\infty},
\end{equation}
where $\lambda$ is the oversampling factor.  
Thus, as long as $\|u\|_{\infty}$ is bounded by an absolute constant on $\mathcal{Y}$, the higher the order of the scheme, the faster the error decay for increasing $\lambda$.\\
For reasons of practicality, one aims for this equality to describe a recurrence relation describing an online procedure and hence seeks to choose $q_n$ in a way that only depends on $u_k$ with $k<n$ and $y_k$ with $k\le n.$ In this paper, in line with previous works, we consider a commonly used subclass of these potential quantization rules by assuming explicit dependence of $q_n$ on only the last $L$ state variables and the current sample. That is, one seeks a quantization function $\mathcal{Q}:\mathbb{R}^{L+1} \to \{-1,1\}$, for which one sets  
\begin{equation}
    q_n := \mathcal{Q}(u_{n-1}, u_{n-2}, \dots, u_{n-L}, y_n).
\end{equation}
In this case, the recurrence relation
\begin{equation*}
    u_n = \sum_{k = 1}^{r} (-1)^{k-1} 
    \left(\begin{matrix}
        r\\
        k
    \end{matrix}\right) 
    u_{n-k} + y_n - q_n,
\end{equation*}
equivalent to \eqref{standardr}, 
%As long as the choice of $q_n$ does not depend on any $u_k$ with $k\ge n,$ this equality is equivalent to the recurrence relation 
uniquely describes the dynamics of the state sequence $(u_k)_k$ starting from given initial values $u_{-L},\dots,u_{-1}$. \\
A seemingly natural approach to choose $q_n$ in a way that, in every iteration, the magnitude of the current state variable is minimized. This corresponds to the \emph{greedy quantization rule}, given by
\begin{equation}\label{greedy_rule}
    q_n = \text{sign}\left(\sum_{k=1}^{r} (-1)^{k-1} \binom{r}{k}\, u_{n-k} + y_n\right),
\end{equation}
where we use the convention $\text{sign}(0) = 1$, that is, 
\[
\operatorname{sign}(x):=
    \begin{cases}
        1, & \text{if } x \ge 0,\\
        -1, & \text{if } x < 0.
    \end{cases}
\]

For $\Sigma\Delta$ schemes of orders of three or higher, however, this greedy approach typically does not yield stability.
Even in the second-order case, while the greedy quantization rule empirically seems to yield stability for all bandlimited signals sampled at sufficiently high rate, this has only been proven for constant inputs below a certain threshold \cite{farrell2002bounding}.
The proof is based on the observation that when the state variable is large, $\Sigma\Delta$ schemes with the greedy quantizer and constant input can be viewed locally as a linear dynamical system and its local parabolic trajectory can be explicitly computed. 
To establish stability for second-order $\Sigma\Delta$ schemes under more general bandlimited inputs, Y\i lmaz \cite{yilmaz2002stability} considers variations of the greedy quantization rule obtained by modifying the coefficients, more concretely
\begin{equation}\label{Yilmaz_quantization_rule}
    q_n = \operatorname{sign}(\gamma u_{n-1} + \Delta u_{n-1}).
\end{equation}

He shows that one can achieve stability for all signals of amplitude bounded by some $\alpha < 1,$ by choosing an appropriate parameter $\gamma = \gamma(\alpha).$ 
More precisely he shows that the state variable can be bounded by
\[
    2\,\frac{1+\alpha}{1-\alpha}.
\]
Daubechies and DeVore \cite{daubechies2003approximating} generalized this idea to higher orders via a quantizer based on nested sign operations.

 %To address this issue, the first family of stable $1$-bit $r$-th order $\Sigma\Delta$ schemes was introduced in [], where it was shown that, if the quantization rule is defined via appropriately constructed nested sign operations, the resulting $\Sigma\Delta$ schemes are stable.  
 %From a mathematical perspective, these quantizers perform very well, but this typically comes at the price of increased implementation complexity, which makes such architectures less attractive in engineering practice, or large bounds on 

%  \begin{figure}
%     \centering
%     \includegraphics[width=0.9\linewidth]{immagini/comparison_quantization rules.jpg}
%     \caption{{\small State variables associated with the quantization of the function $f(x) = 0.95 \,\mathrm{sinc}(x)$ for $x\in [-3,3]$ under different quantization rules. The first one is the greedy quantization rule, the second one is the one proposed by Y\i lmaz \cite{yilmaz2002stability}. Observe that the maximum amplitude is $\alpha = 0.95.$}}
%     \label{Amplitude_too_large}
% \end{figure}

\begin{figure}
    \centering
    \includegraphics[width=1\linewidth]{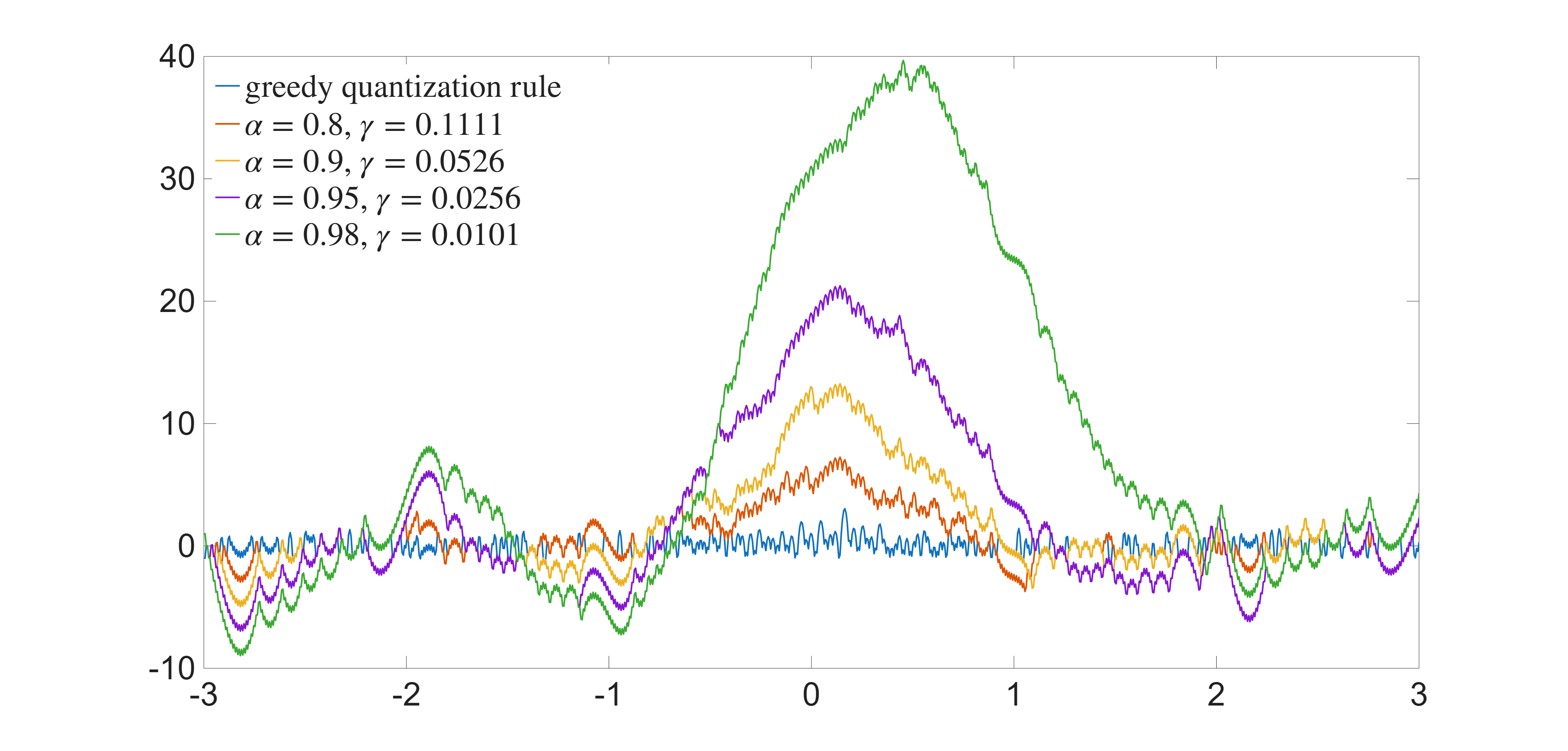}
    \caption{{\small State variables associated with the quantization of the function $f(x)=0.5\,\mathrm{sinc}(x)$ for $x\in[-3,3]$, using the greedy quantization rule--for which no theoretical guarantees are available for general bandlimited functions--and the different quantization rules proposed by Y\i lmaz~\cite{yilmaz2002stability}.}}
    \label{Yilmaz_plotz}
\end{figure}

A drawback of these approaches is that the parameter choice and the resulting bound depend on the maximal admissible amplitude, rather than the actual amplitude. Indeed, Figure~\ref{Yilmaz_plotz} demonstrates that for a fixed signal, increasing the admissibility parameter $\alpha$ in the second-order scheme proposed by Y\i lmaz directly translates into an increased amplitude of the state variable. This is not desirable: it has been demonstrated in a number of works that a large magnitude of the state variable is directly
associated with a degradation of the reconstructed signal (see \cite{kwon2004op,chen2009modeling,beilleau2003systematic}).
For this reason, the greedy quantization rule often remains the method of choice even when its stability is only empirically observed and no guarantees are available, given that it
minimizes the magnitude of the state variable at each iteration.
In Figure \ref{Yilmaz_plotz}, it can be seen empirically that the maximum amplitude of the state variable for the greedy quantizer is smaller than for any of the quantizers proposed by Y\i lmaz. 

This naturally raises the question of whether it is possible
to prove more general stability guarantees for higher-order $\Sigma\Delta$ schemes with the greedy quantization rule.

\subsection{Generalized \texorpdfstring{$\Sigma\Delta$}{Sigma-Delta} schemes}
A common remedy for establishing stability beyond the standard first and second order schemes discussed above is to introduce a carefully designed feedback filter in the recurrence relation. This idea has been suggested as well as empirically explored in \cite{schreier2005understanding} and is commonly used in practice.
In mathematical terms this corresponds to a generalized $\Sigma\Delta$ schemes in the sense of the following definition.
\begin{definition}
Let $y = (y_n)_{n\in \mathbb{N}_0} \subset \mathbb{R}$ be a sequence with $\|y\|_{\infty}\le 1$.  
A generalized $1$-bit $\Sigma\Delta$ scheme is defined by initial data $(v_n)_{n\in \mathbb{Z}_{-}}$ -- a common choice would be $v_n=0$ for all negative integers $n$ -- and a recurrence relation of the form
\begin{equation}\label{generalized_SD}
\left\{
        \begin{aligned}
            q_n &= \operatorname{sign}\bigl((h\ast v)_n + y_n\bigr),\\
            v_n &= (h\ast v)_n + y_n - q_n,
        \end{aligned}
\right.
\end{equation}
where $h \in \ell^1(\mathbb{N}_0)$ is a feedback filter with $h_0 = 0$ 
and for all $n\in \mathbb{N}_0$
\[
    (h\ast v)_n = \sum_{i = 1}^{\infty} h_i\, v_{n-i}.
\]
\end{definition}

While the definition above is stated for an arbitrary filter $h$, only suitable feedback filters produce the desired noise-shaping effect.  
The following structural condition on $h$ guarantees noise shaping, provided that the scheme is stable.

\begin{definition}
A $\Sigma\Delta$ scheme as in \eqref{generalized_SD} is said to be of $r$-th order if 
\begin{equation}\label{moment_condition}
    \delta_0 - h = \Delta^r g,
\end{equation}
for some $g\in \ell^1(\mathbb{N}_0)$, where $\delta_0$ denotes the Kronecker delta at $0$.
\end{definition}

In this paper, we focus on the case where the feedback filter $h$ is of finite length $L$, that is, $h_j = 0$ for all $j>L$, a so-called \emph{finite impulse response} (FIR) filter. While there are also examples of $\Sigma\Delta$ schemes based on infinite impulse response filters \cite{johns2002design}, FIR filters are commonly used and have also been the focus of prior theoretical investigations.\\
Finite length of the feedback filter implies that also the initial condition can be described by finitely many values, namely $v_{-L},\dots, v_{-1}.$ Furthermore, for filters of length $L,$ condition \eqref{moment_condition} is equivalent to the equalities 
\[
    \sum_{i=0}^L (\delta_0^i - h_i)\, i^k = 0 
    \qquad \text{for all } k = 0,\dots, r-1,
\]
which are commonly referred to as the \emph{moment conditions}.\\
%Moreover, choosing $g = \delta_0$ recovers the standard $\Sigma\Delta$ schemes.  
Through a reparametrization, an $r$-th order generalized $\Sigma\Delta$ scheme \eqref{generalized_SD} can be expressed as a standard $r$-th order scheme but no longer with the greedy quantization rule, see \cite{gunturk2003one}. Indeed, if we define $u := v \ast g$, then the generalized $r$-th order scheme~\eqref{generalized_SD} can be expressed as
\[
    \Delta^r u = y - q,
\]
so $u$ satisfies the defining inequality for standard $r$-th order $\Sigma\Delta$ quantizer. A direct implementation of this recurrence relation, however, is not possible as the sequence $q$ is unchanged and hence still computed based on $v$ and not on $u.$ Nevertheless this reparametrization is crucial for the theoretical analysis, as the reconstruction error of a standard $r$-th order $\Sigma\Delta$ scheme can be bounded in terms of $\|u\|_{\infty}$ via \eqref{error_estimate}.
To bound $\|u\|_{\infty}$ one can apply Young's inequality to the defining equation $u = v\ast g$ and obtain
\begin{equation}
    \|u\|_\infty \le \|v\|_\infty\, \|g\|_{\ell^1},
\end{equation}
which yields the error bound 
\begin{equation}\label{error_estimate_v}
    \|f - f_q\|_{\infty}
    \;\le\;
    \frac{1}{\lambda^r}\,\|\varphi^{(r)}\|_{L^1}\,\|v\|_{\infty}\|g\|_{\ell^1}.
\end{equation}

Hence, it remains to find a filter $h$ that satisfies the moment conditions of order $r$ (and the corresponding $g$) such that the associated generalized $\Sigma\Delta$ scheme is stable for inputs of interest.
For input sequences bounded by a small constant, a sufficient condition is given by the following proposition, established in \cite{anastassiou2002error,schreier1991stability}.
We include the proof for the readers' convenience.

\begin{prop}\label{lemma_classical_criterion}
    Let $y = (y_n)_n$, with $\|y\|_{\infty}\le 1$ and a $\Sigma\Delta$ scheme as in \eqref{generalized_SD}.
    If \begin{equation}\label{classical_criterion}
        \|h\|_{\ell^1} + \|y\|_{\infty}\le 2,
    \end{equation}
    then the scheme is stable for any  initial values $v_{-L},\dots,v_{-1}\in [-1,1]$ and one has $\|v\|_{\infty}\le 1.$
\end{prop}
\begin{proof}
    By induction. For the initial values, the statement holds by assumption.\\ 
    For the induction step, we assume that for all $k = n-L,\dots, n-1,$ we have $v_k\in [-1,1].$\\
    We first consider the case that $(h\ast v)_n +y_n\ge 0$ and then $q_n = 1.$
    By Young's inequality and assuming that $|v_k|\le 1$ for all $k = 0,\dots,n-1$ we have $|(h\ast v)_n|\le \|h\|_{\ell^1},$ therefore
    \begin{equation*}
        v_n = (h\ast v)_n +y_n - q_n\le \|h\|_{\ell^1} + \|y\|_{\infty} - 1\le 1.
    \end{equation*}
    Furthermore, $-v_n = -(h\ast v)_n -y_n + q_n\le q_n = 1,$ so we have
    $|v_n|\le 1$.\\
    For  $(h\ast v)_n +y_n\le 0$, the proof proceeds analogously.
\end{proof}

G\"unt\"urk \cite{gunturk2003one} proposed to exploit this proposition together with the error bound \eqref{error_estimate_v} to design $\Sigma\Delta$ schemes that are stable for all input sequences $y$ of a given maximum amplitude $\mu$.\\
As the bound \eqref{error_estimate_v} gets stronger the smaller $\|g\|_{\ell^1}$ is, this strategy boils down to finding $r$-th order filters for which $\|g\|_{\ell^1}$ is as small as possible while at the same time the stability criterion\eqref{classical_criterion} is satisfied. Under the aforementioned length constraint this gives rise to the optimization problem
\begin{equation}\label{minimization_g_l1}
    \operatorname*{argmin}\|g\|_{\ell^1} 
    \quad \text{s.t.} \quad 
    \begin{cases}
        \displaystyle \sum_{i = 0}^L (\delta_0^i - h_i)i^k = 0, & \text{for all } k = 0,1,\dots, r-1 \\[3pt]
        \|h\|_{\ell^1}\le 2-\mu,
    \end{cases}
\end{equation}

%In real application, given the maximum amplitude $\mu$ of a signal, and using the classical criterion as a guideline for achieving stable quantization, one aims at finding the shortest filter $h$ with minimal $\|h\|_{\ell^1}$ such that

In numerical simulations, one consistently observes that the filters $h$ corresponding to the solutions of this minimization problem are sparse, for suitable parameters $\mu$ even minimally sparse in the sense of the following definition.

%, for a fixed filter length $L$, when solving the optimization problem
%\begin{equation}\label{minimization_h_l1}
%    \operatorname*{argmin}\|h\|_{\ell^1} 
%    \quad \text{s.t.} \quad 
% \begin{cases}
% \displaystyle \sum_{i = 0}^L (\delta_0^i - h_i)i^k = 0, & \text{for all } k = 0,1,\\[3pt]
% \|h\|_{\ell^1}\le 2-\mu,
% \end{cases}
% % \end{equation}
% one obtains so-called filters with minimal support, which are the main focus of [] and [] and for which we now give a precise definition.

\begin{definition}
A $r$-th order $\Sigma\Delta$ feedback filter $h$ is said to be of \emph{minimal support} if
\[
    h = \sum_{j = 1}^{r} d_j \,\delta^{(n_j)},
\]
with $\{d_j\}_{j = 1}^r \subset \mathbb{R}$ and $1 \le n_1 < n_2 < \dots < n_r$.
\end{definition}

This observation motivated the study of the optimization problem  \eqref{minimization_g_l1} over the class of filters $h$ of minimal support. A first step was taken in \cite{gunturk2003one} by providing for each $r\in \mathbb{N}$ and $\mu \in (0,1)$ a construction of a minimally supported $r$-th order filter $h$ satisfying $\|h\|_{\ell^1}\le 2-\mu$ and explicitly computing the corresponding $\|g\|_{\ell^1}.$
Choosing the order $r$ appropriately as a function of $\lambda$, this construction was shown to give rise to an error bound decaying as 
\[
    \mathcal{O}\bigl(e^{-c\lambda}\bigr)
\]
for some $0 < c < 1.$ This bound is near-optimal in the sense that a lower bound of the same form (with different but also amplitude dependent constant $c$) was established in \cite{krahmer2012lower}. Subsequently, the analysis of \cite{gunturk2003one} was refined by explicitly solving a relaxed version of the optimization problem \eqref{minimization_g_l1} over minimally supported filters and discretizing the solution in an asymptotically optimal sense \cite{deift2011optimal}. 

In this paper, we are focusing on the case of second-order, $r=2;$ then the resulting minimally supported filters are of the form 
\begin{gather}\label{minimal_filter}
h = \underbrace{\left(\frac{k+1}{k}, 0, \ldots, 0, -\frac{1}{k}\right)}_{k+1} \quad \text{for } k \geq 1,
\end{gather}
and thus satisfy \( \|h\|_{\ell^1}  =  \frac{k+2}{k} \).
Therefore, for a given $\mu,$ the stability condition translates into considering the filter with minimal $\ell^1$ norm, with $k$ such that 
\begin{equation}
    k\ge \frac{2}{1-\mu},
\end{equation}
in other words, the length of the filter scales as $\mathcal{O}\left(\frac{1}{1-\mu}\right).$

For $\mu$ close to $1$ this can lead to very long filters, e.g., if one aims to admit signals up to an amplitude of $0.98,$ one needs a filter of length $100,$ which is not desirable from the viewpoint of circuit complexity. Similarly to the parameter $\alpha$ discussed above, this length only depends on the maximal admissible amplitude, and one does not gain if the actual amplitude is smaller.
At the same time, $\Sigma\Delta$ has been empirically observed to be stable considerably beyond the criterion \ref{classical_criterion} \cite{schreier2005understanding}, see Section \ref{sec_empiric} for an empirical exploration of the stability of filters of the form \eqref{minimal_filter}. For this reason, it is common to use $\Sigma\Delta$ schemes in regimes where their stability is not theoretically guaranteed by the current literature, see \cite{schreier2005understanding} for details.\\
In practice, it is very common to use the following heuristic approach: a filter design is empirically tested for stability under constant inputs, and once stability is verified in this setting, one applies the scheme for more general signals, expecting that it remains stable for a sufficiently high sampling rate.

However, to the best of our knowledge, this intuition has never been rigorously verified, not even for second-order schemes.  
This gap in the theory forms the main motivation for this paper.

Indeed, we show that the stability condition for minimally supported second-order filters can be substantially improved. 
More precisely, we establish the following theorem (here given in a simplified form, see Theorem \ref{Main_Theorem} for the more quantitative version).

\begin{theorem}
Let $h$ be a second order filter with $k\ge 3,$ and $f\in PW_{1}(\mathbb{R}),$ such that $\|f\|_{\infty}< 1- \frac{e^2}{(k+1)^2}$. There exist $c:= c(\|f\|_{\infty},k)>0$ and $C:= C(\|f\|_{\infty},k)>0$, monotonically increasing in the first variable and decreasing in the second, such that if $f$ is quantized with oversampling $\lambda\ge c$ and with initial conditions in $[-1,1]^{k+1},$ then the scheme is stable and 
  \vspace{-.2cm}
  \[
\|v\|_{\infty}\le C(\|f\|_{\infty},k).
    \]
\end{theorem}
This result provides a significant improvement with respect to Proposition \ref{lemma_classical_criterion}: the bound establishes stability for filters as short as  $\mathcal{O}\left(\frac{1}{\sqrt{1-\mu}}\right)$ and bandlimited inputs sampled at a high enough rate. For example, for quantizing a function of amplitude $0.98$, one would need a filter of only length $18$ under sufficient oversampling.

The key technical innovation underlying our result is that for the first time our result explicitly exploits the smoothness of the bandlimited signal, while most previous stability analyses have been conducted for arbitrary input sequences, without fully considering that these sequences originate from samples of bandlimited functions.

%This result suggests that by incorporating regularity information on the input signal, one can design shorter, more efficient filters that maintain global stability even at high input amplitudes.

\section{Empirical stability and instability beyond the \texorpdfstring{$\ell^1$}{l1} guarantee}\label{sec_empiric}
In this section we will demonstrate empirically that stability is closely tied to the smoothness of the underlying signals. We first illustrate that for the very smooth signal $f(x) = 0.785\sin(4\pi x)$ sampled at rate $\lambda = 500$, stability is observed beyond the classical criterion, see Figure~\ref{fig:v_n_in_Gamma_+}).

%For time-varying signals, stability for the greedy quantization rule has been ensured only by resorting to longer second-order filters and applying the classical criterion \eqref{classical_criterion}.  
%As discussed earlier, however, this criterion is overly conservative and often fails to reflect the empirically observed stable behavior of the system when it is violated.  

\begin{figure}[h]
    \centering
    \includegraphics[width=1\linewidth]{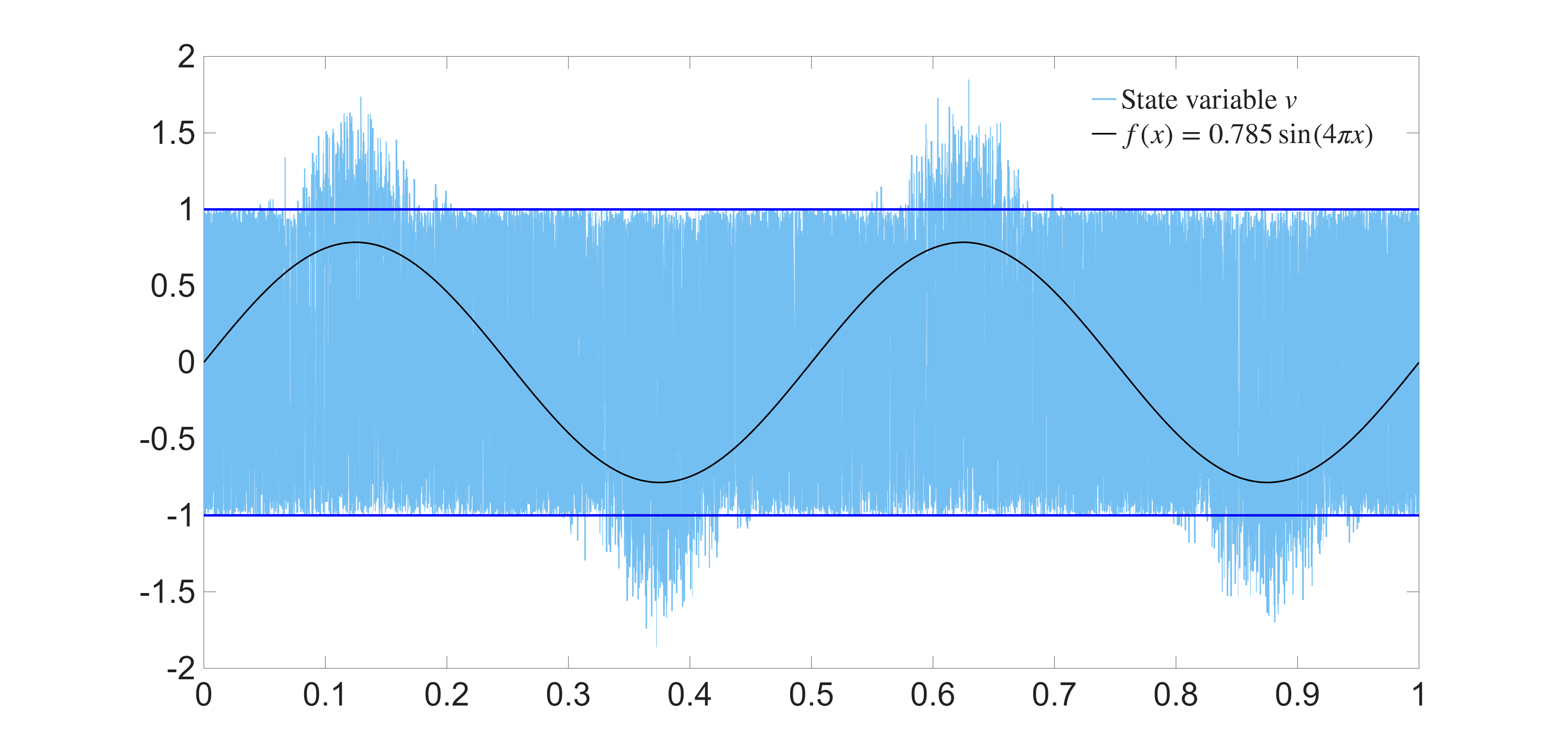}
    \caption{State variables $v = (v_n)_{n}$ generated by the function $f(x) = 0.785\sin(4\pi x)$ in $[0,1]$ and filter $h$ as in \eqref{minimal_filter} with $ k = 3$. It's also worth noting that the state variable exceeds $1$ or $-1$ only in the vicinity of the extreme function values.}
    \label{fig:v_n_in_Gamma_+}
\end{figure}
The figure also demonstrates that the proof technique of Proposition~\eqref{lemma_classical_criterion} can no longer be applied, as the induction hypothesis that the state variable remains bounded by $1$ is clearly violated. 

%The main issue that arises when $\|y\|_{\ell^{\infty}} > 1-\tfrac{2}{k}$ is that the induction argument used in Lemma, since some of the values $v_n$ may fall outside the interval $[-1,1]$ (see Figure~\ref{fig:v_n_in_Gamma_+}).  
%Consequently, determining whether a uniform bound exists for the sequence $(v_n)_n$ becomes nontrivial, even though in practice the quantization process often appears to remain stable for many input signals.

When the samples do not arise from a smooth signals or are not obtained at a high enough rate, the same filter and the same amplitude as before can give rise to instabilities. To construct a counterexample, we follow the idea of \cite{yilmaz2002stability} and construct sequences $(y_n)_{n\in\mathbb{N}},$ with $\|y\|_{\infty}<0.8,$ for which $|v_n| \to \infty$ as $n \to \infty$, showing that the scheme may fail to keep the internal state variables bounded, see Figure~\ref{Critical_scenarios}.  

\begin{figure}[h!]
  \noindent\makebox[\linewidth][c]{%
    \resizebox{1.2\linewidth}{!}{%
      \begin{subfigure}{0.5\linewidth}
        \centering
        \includegraphics[width=\linewidth]{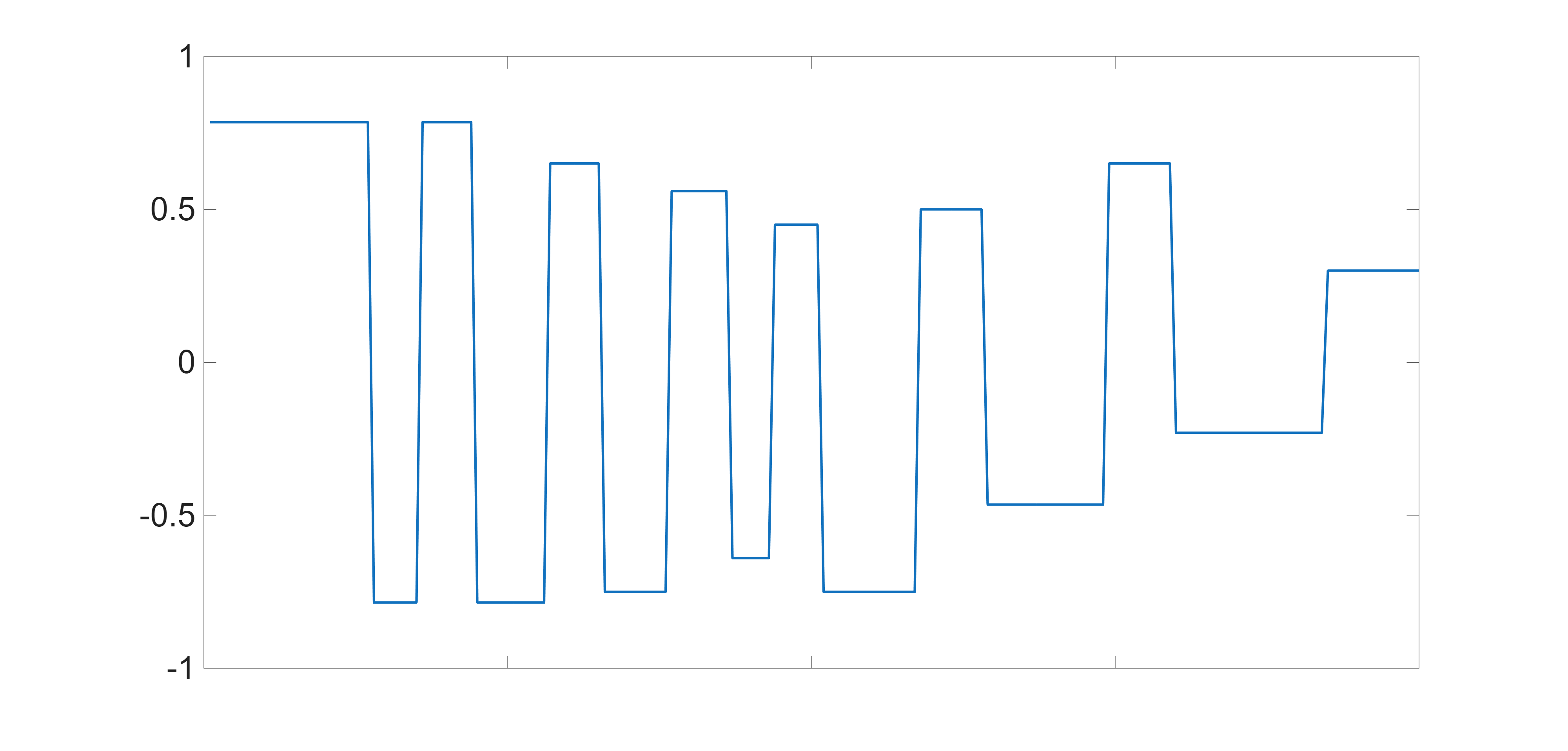}
        \caption{{\tiny Plot of the step function that triggers instability.}}
        \label{Step_function}
      \end{subfigure}
      \begin{subfigure}{0.5\linewidth}
        \centering
        \includegraphics[width=\linewidth]{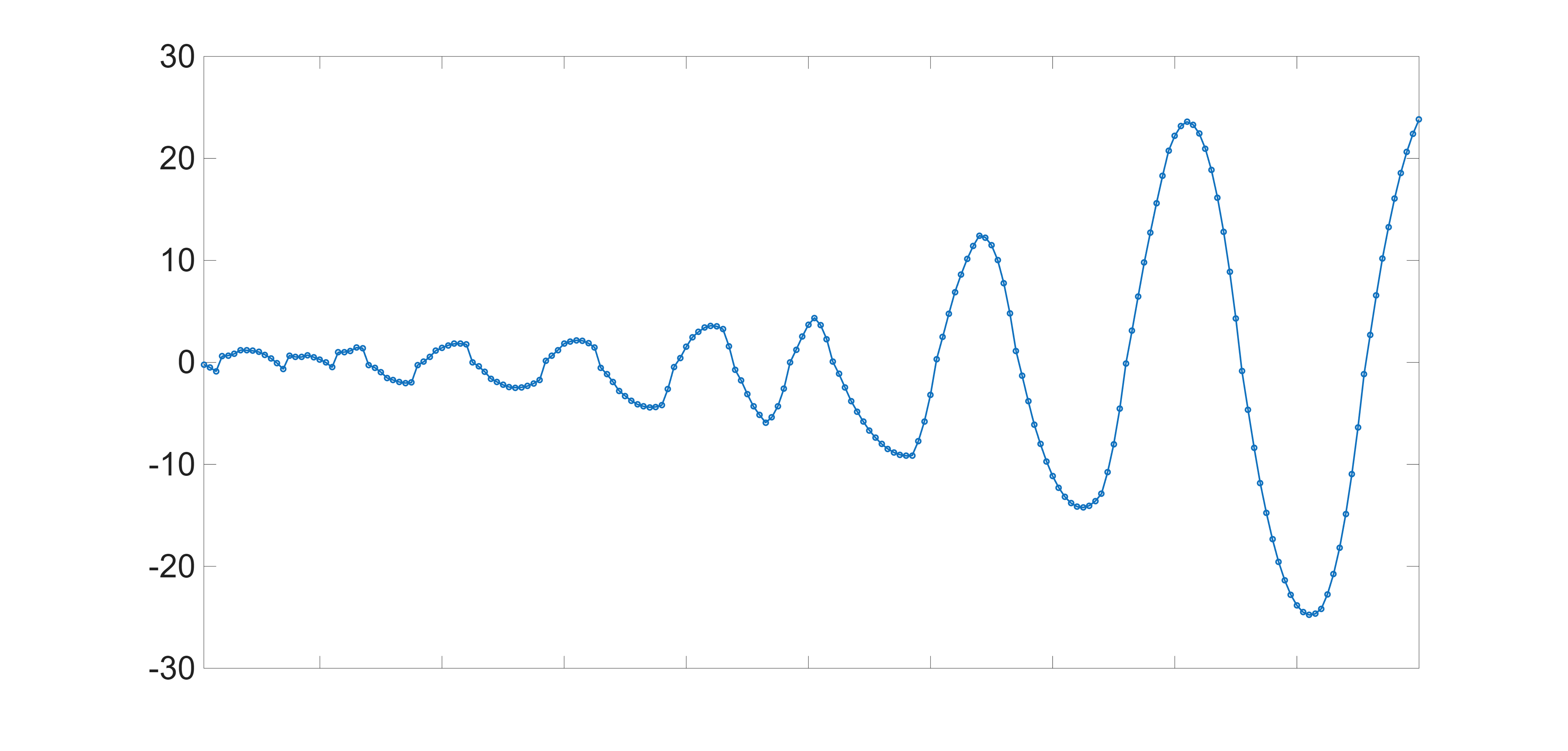}
        \caption{{\tiny State variables obtained by quantizing the signal in \ref{Step_function}.}}
        \label{Oscillating_variable}
      \end{subfigure}
    }%
  }
  \caption{Critical type of scenarios that lead to instability: example with  filter $h$ as in \eqref{minimal_filter} and $k =3$.}
  \label{Critical_scenarios}
\end{figure}
The idea is to design the sample sequence so that positive and negative sample values align with regions where the state variable increases or decays, respectively. Hence, this design crucially incorporates rapid changes of the sample values at the maxima and minima of the state variables. These changes then lead to an increase in the growth or decay rate, giving rise to overshoot effects that gradually increase the amplitudes of the state variable oscillations. Analyzing how the smoothness of the sample sequence can prevent these effects is a main contribution of this paper.

\section{Stability Improved Second Order Filters}
In this section, we develop the theoretical framework needed to address the stability problem described above.  
The analysis is organized into three parts.  
First, we analyze some structural properties of $\Sigma\Delta$ schemes that will help us understand their stability.  
Second, we investigate the stability of the scheme for constant input sequences in the regime where 
\eqref{classical_criterion} is not satisfied.  
Finally, we show how these results extend from constant signals to bandlimited signals.

\subsection{Problem setup and key concepts}
%Before turning to the details, we introduce some definitions, basic properties, and preliminary observations that will guide the analysis and clarify the strategy for the aforementioned problem.
As discussed in the preliminaries, we focus on second-order $\Sigma\Delta$ schemes equipped with feedback filters of minimal support, namely
\begin{gather}\label{minimal_filter2}
    h = \underbrace{\left(\tfrac{k+1}{k},\, 0,\, \ldots,\, 0,\,-\tfrac{1}{k}\right)}_{k+1},
\end{gather}
and we study the corresponding recurrence relation 
\begin{equation}\label{second_order_filter}
\left\{
\begin{aligned}
    v_n &= \frac{k+1}{k}\, v_{n-1} \;-\; \frac{1}{k}\, v_{n-(k+1)} \;+\; y_n \;-\; q_n, \\[2mm]
    q_n &= \operatorname{sign}\!\left( \frac{k+1}{k}\, v_{n-1} \;-\; \frac{1}{k}\, v_{n-(k+1)} \;+\; y_n \right),
\end{aligned}
\right.
\end{equation}
where $y = (y_n)_{n \in \mathbb{N}_0}$ is a given input sequence and 
\[
\operatorname{sign}(x) :=
\begin{cases}
    1, & x \ge 0, \\
    -1, & x < 0.
\end{cases}
\]
A well-known key observation for our analysis (see, e.g., \cite{yilmaz2002stability}) is that for large or large negative values of the state variable, its sign agrees with the quantized sample, as summarized in the following lemma. 
\begin{lemma}\label{u_ningamma+}
    Given a $\Sigma\Delta$ scheme as in \eqref{second_order_filter}, the following properties hold. 
    \begin{itemize}
        \item[i)]  If $v_n\in (1,+\infty),$ then $q_n = 1.$
        \item[ii)]  If $v_n \in (-\infty,-1),$ then $q_n = -1.$
        \item[iii)] If $q_n = 1,$ then $v_n \ge -1.$
        \item[iv)] If $q_n = -1,$ then $v_n < 1.$
    \end{itemize}
\end{lemma}
\begin{proof} 
    (i) Observe that if a sequence $(a_n)_{n\in\mathbb{N}}$ satisfies
    \begin{equation}
        v_n = a_n -\text{sign}(a_n) \label{eq:av}
    \end{equation}
    then, if $v_n > 1$ one has
    \begin{equation*}
        a_n > 1+ \text{sign}(a_n)\ge 0,
    \end{equation*}
    therefore $\text{sign}(a_n) = 1.$ The statement follows by observing that for the given $\Sigma\Delta$ schemes, $a_n := (h\ast v)_n+y_n$ satisfies \eqref{eq:av}  because  $q_n = \text{sign}(a_n).$\\\\
    ii) Same as in i).\\
    iii) Negation of statement ii)\\
    iv) Negation of statement i).    
\end{proof}

As noted in \cite{yilmaz2002stability}, this observation yields that a $\Sigma\Delta$ scheme equipped with the greedy quantization rule can be viewed as a piecewise-defined dynamical system.  
In particular, according to Lemma~\ref{u_ningamma+}, when the state variable lie in 
$(1,+\infty)$ (resp.\ $(-\infty,-1)$) they evolve according to the linear dynamics
\begin{equation}\label{linear_system}
    v_n = (h \ast v)_n + y_n - 1 
    \qquad \text{(resp.\ } v_n = (h \ast v)_n + y_n + 1\text{)}.
\end{equation}
In contrast, whenever a value $v_n$ falls into the interval $[-1,1]$, the quantizer output $q_n$ depends on the history and may change from sample to sample, i.e., the evolution of the state variables is governed by a non-linear dynamical system.  
This distinction is essential, as it allows us to analyze the global behavior of a $\Sigma\Delta$ scheme by separately understanding its linear and non-linear regimes.

Since our goal is to fully characterize the scenarios in which the state variables enter either $(1,+\infty)$ or $(-\infty,-1)$, we must understand which configurations of consecutive state variables in $[-1,1]$ may trigger such transitions.  
We formalize this notion in the following definition. 

\begin{definition}
    Fix a $\Sigma\Delta$ scheme as in \eqref{generalized_SD}, with a filter $h$ of length $L$. 
    \begin{enumerate}
        \item[a)]  We say that a vector $\bar{v} =(v_{-L},\dots,v_{-1})\in [-1,1]^L$ is a trigger of reach $m\ge 0$ for an input sequence $y = (y_n)_{n\in\mathbb{N}_0}$ if the $\Sigma\Delta$ scheme with initial values $\bar{v}$ applied to the input sequence yields $(v_0,\dots,v_m)\in (1,+\infty)^{m+1}$ (or $(-\infty,-1)^{m+1}$). In the former case, we speak of a positive trigger; in the latter case, of a negative trigger.
        \item[b)] A generalized positive triggering sequence is defined in the same way as a positive trigger, except that we also allow for larger entries, concretely $\bar{v} \in [-1,+\infty)$,  $v_{-1}\in [-1,1]$ (analogously for generalized negative triggering sequence).
        \item[c)] In all these cases, we say that $(v_0,\dots,v_m)$ is a trajectory associated to $\bar{v}.$ Moreover, we define the coverage $N(\bar{v},y)$ of a positive trigger or generalized positive triggering sequence as
    \begin{equation}
        N(\bar{v},y):= \text{sup}\{m\in \mathbb{N} \ | \ v_j\in (1,+\infty)\ \ \text{for all} \ 0\le j\le m\}.
    \end{equation}
     Similarly, the definition extends to negative triggering sequences.
    \end{enumerate}
\end{definition}

The following lemma states that for constant input sequences, the sign of a trajectory always agrees with the sign of the input $y$.

\begin{lemma}\label{where_the_trajectory_starts}
Let $k \ge 2$.
If $y > 0$ (or $y < 0$, respectively), then the trajectory generated for the constant input sequence $y_n = y$ by
any trigger $\bar{v}$ of reach $m$ lies in $(1,+\infty)^m$ (or $(-\infty,-1)^m$, respectively).
\end{lemma}

\begin{proof}
If $|y|\le 1-\frac2k$ then by Proposition \ref{classical_criterion} one has $\|v\|_{\infty}\le 1$, so there are no triggers. Hence, we need to consider the cases $y>1-\frac2k$ and $y<-\left(1-\frac2k\right)$. 

We first consider the case $y>1-\frac2k$. Let $\bar{v} = (v_{-(k+1)},\dots,v_{-1}) \in [-1,1]^{k+1}$ be a trigger.
Assume by contradiction that the first iterate $v_0$ lies in $(-\infty,-1)$.
By Lemma~\ref{u_ningamma+}, we have
\begin{align*}
    v_0
    &= \frac{k+1}{k}v_{-1} - \frac{1}{k}v_{-(k+1)} + y + 1 \\[0.5em]
    &> -\frac{k+2}{k} + \frac{k-2}{k} + 1
      = 1 - \frac{4}{k} \;\;\ge\; -1,
\end{align*}
where we used $v_{-1}, v_{-(k+1)} \in [-1,1]$ and $y > 1 - \tfrac{2}{k}$.
This contradicts the assumption $v_0 < -1$, proving the claim.
The case $y < \tfrac{2}{k}-1$ follows analogously with reversed signs.
\end{proof}

\begin{remark}
    The constraint $k \ge 2$ is crucial here; for $k = 1,$ the standard second-order filter, one can also have negative trajectories for positive inputs.
\end{remark}
As a consequence of this lemma, when considering constant input sequences, one only needs to consider the cases that $y\ge0$ and the trajectory lies in $(1,+\infty)^m$ and that $y\le0$ and the trajectory lies in $(-\infty,-1)^m$.
In the following, we will focus on the first case; by symmetry of the quantization alphabet $\{-1,1\}$, all results extend directly
to the second case.

\subsection{Trajectories for constant input sequences}
The goal of this section is to understand the (positive) trajectories generated for positive constant input sequences $y_n=y>0$ and their triggering mechanisms. By Lemma~\ref{u_ningamma+}, one has $q_n=1$ for all $n$ corresponding to the trajectory elements, so the trajectory is described by a linear dynamical system. One numerically observes that it displays a parabolic-type behavior, analogous to the standard second-order case (see Figure \ref{Constant_input_image}). The goal of this section is to mathematically describe this behavior.

\begin{figure}
    \centering
    \includegraphics[width=0.9\linewidth]{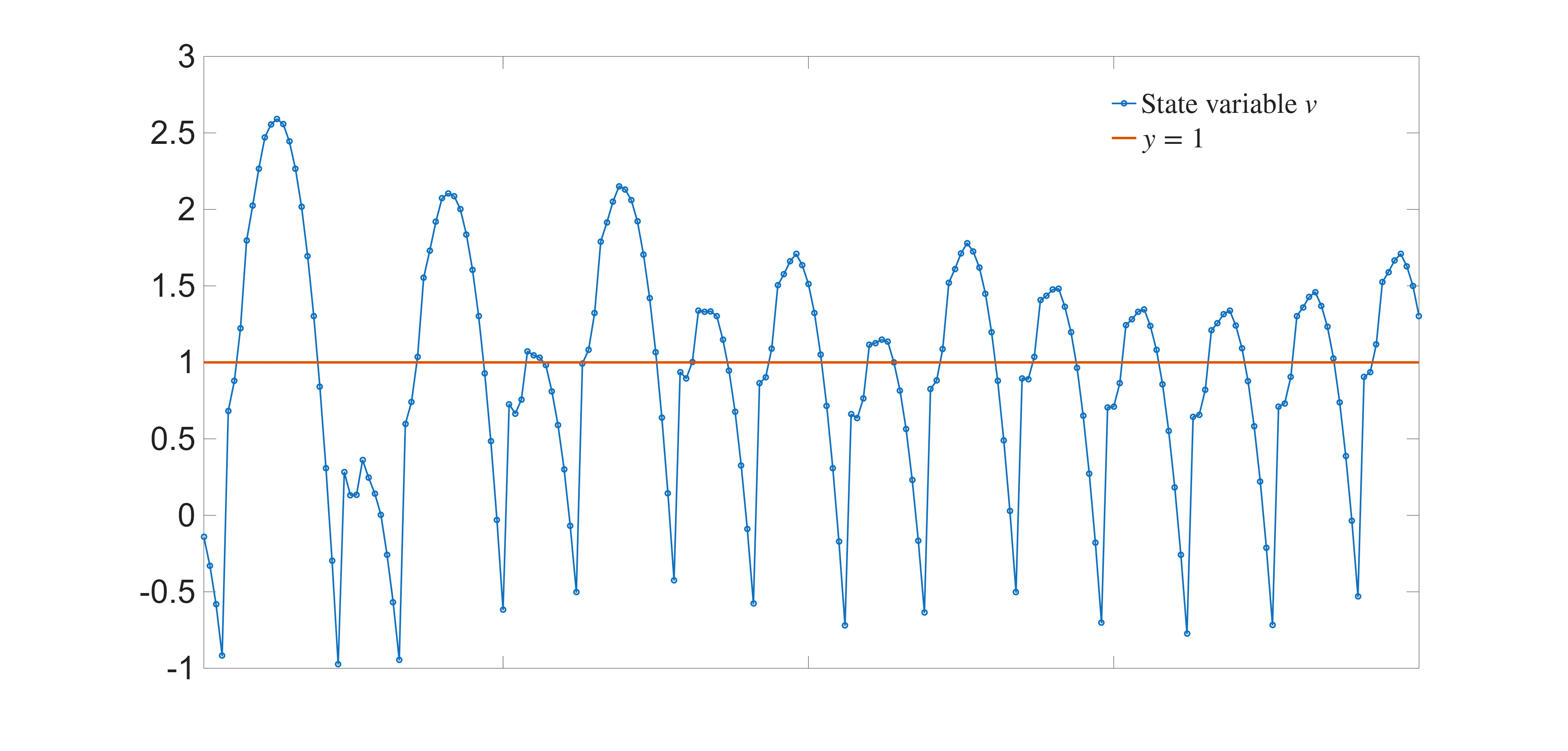}
    \caption{Behavior of the state variable when quantizing the constant signal
    $y = 0.85845$ using the filter $h$ as in \eqref{minimal_filter} and initial condition
    $\tilde{v} = (0,0,0,0)$.}
    \label{Constant_input_image}
\end{figure}

% A key observation is that, for filters of arbitrary length, whenever
% \[
%     \|h\|_{\ell^1} + \|y\|_{\ell^\infty} > 2,
% \]
% the trajectories entering the regions $(1,+\infty)$ or $(-\infty,-1)$ display a
% parabolic-type behavior, analogous to that of the standard second-order
% $\Sigma\Delta$ scheme; see Figure~\ref{Constant_input_image}.
% In this regime, 

% This fact allows us to restrict the analysis to the associated linear dynamical
% system.
% In particular, we investigate how triggering events determine the geometry of the
% resulting trajectories, as this information is essential for understanding
% stability beyond the classical criterion.

% Lemma~\ref{where_the_trajectory_starts} shows that, for constant inputs exceeding
% the threshold $1-\tfrac{2}{k}$ in magnitude, any trajectory generated by a trigger
% immediately can only enter $(1,+\infty)$.

A central challenge in analyzing these trajectories is that, for schemes that differ from the standard case, the trajectory's shape is not exactly parabolic, due to the particular configuration of the trigger that generates it. To better understand these trajectories, we adopt the common strategy in linear dynamical systems theory of tracing back the trajectory elements to the initial value.
More precisely, let $\bar{v} = (v_{-(k+1)},\dots,v_{-1})$ be a positive trigger for $y>0$ with coverage $N(\bar{v},y),$ and $H$ be the matrix of dimension $(k+1)\times(k+1)$ given by 
\begin{gather}
    H = 
    \left(
\begin{matrix}
\frac{k+1}{k} & 0 & \cdots & 0 & -\frac{1}{k} \\
1 & 0 & \cdots & 0 & 0 \\
0 & 1 & \ddots & \vdots & \vdots \\
\vdots & \ddots & \ddots & 0 & 0 \\
0 & \cdots & 0 & 1 & 0
\end{matrix}
\right),
\end{gather}

then for all $n = 0\,\dots,N(\bar{v},y)$ one has

\begin{gather}\label{Recursion_matrix}
    v^{(n)} = Hv^{(n-1)} + (y-1)e_1 = H^2 v^{n-2} + (y-1)\sum_{i = 0}^1 H^i e_1 = \dots = H^{n+1}v^{(-1)} +(y-1)\sum_{i = 0}^{n}H^i e_1,
\end{gather}
where 
\begin{gather}
    v^{(n)} = 
    \left(
    \begin{matrix}
        v_{n}\\
        v_{n-1}\\
        \vdots \\
        v_{n-(k+1)}
    \end{matrix}
    \right) \quad \text{and} \quad e_{1} = \left(
    \begin{matrix}
        1\\
        0\\
        \vdots \\
        0
    \end{matrix}
    \right). 
\end{gather}
In this way, by analyzing the coefficients of the matrix $H^n$, we can understand how the state elements of the trajectory evolve once the trigger $v^{(-1)}$ is fixed.
In the following proposition we reformulate \eqref{Recursion_matrix} for the elements of the trajectories. This will be essential for obtaining a global upper bound for the triggered trajectories.
We carry out this analysis in the constant-input setting, which will later be shown
to determine the admissible stability margins for bandlimited signals as well.

% A central challenge in analyzing these trajectories is that their shape is not immediately apparent. This is because the trajectory shape strongly depends on the specific trigger configuration that generates them.
% To address this challenge, we adopt a straightforward yet effective strategy: we reformulate the recursion~\eqref{second_order_filter} so that each trajectory element can be expressed as a linear combination of appropriate coefficients and trigger values.

% This new representation clarifies how the trajectory depends on the trigger. In contrast to the classical dynamical-systems approach, which typically relies on spectral analysis—where the influence of the trigger is obscured—our formulation allows us to describe the evolution directly in terms of the quantization step and the trigger configuration.

% This perspective enables us to derive global bounds for the trajectories, which
% are essential for the stability analysis of the scheme.

% For that, the following proposition reformulates the recurrence relation \eqref{minimal_filter} so as to make the dependence on the
% trigger explicit and derive a global upper bound on the resulting trajectories that is therefore
% independent of the specific trigger.
% We carry out this analysis in the constant-input setting, which will later be shown
% to determine the admissible stability margins for bandlimited signals as well.

\begin{prop}\label{evolution_highest_peak}
Let $h$ be a filter as in \eqref{minimal_filter} and let 
$\bar{v} = (v_{-(k+1)},\dots,v_{-1})\in [-1,1]^{k+1}$ be a trigger of reach $M$, with constant input 
\(0 < y < 1\).  
Let \((v_n)_{n = 0,\dots, M} \in (1,+\infty)^{M+1}\) be the trajectory associated with $\bar{v}$.  
Then, for every $n = 0,\dots,M$,
\begin{equation}\label{characterization_of_trajectories}
    v_n 
    = h_1^n v_{-1} + h_2^n v_{-2} + \dots + h_{k+1}^n v_{-(k+1)}
      + (y-1)\sum_{i = 0}^{n} h_1^{i-1},
\end{equation}
where the coefficients $h_i^n$ are defined as follows:
\begin{itemize}
    \item[i)]  for $i = 1,\dots,k+1$
    \begin{equation}\label{expression_coefficients_0}
        h_i^{0} := h_i,
    \end{equation}
    \item[ii)] for $n \ge 1$ and $i = 1,\dots,k-1$,
    \begin{equation}\label{expression_coefficients_1}
        h_i^{n} = h_i \, h_1^{n-1} + h_{i+1}^{n-1};
    \end{equation}
    \item[iii)] for $n \ge 1$,
    \begin{equation}\label{expression_coefficients_2}
        h_{k+1}^{n} = h_{k+1} \, h_1^{n-1}.
    \end{equation}
\end{itemize}
with the convention that $h_1^{-1}:= 1.$
\end{prop}

\begin{proof}
By Lemma~\ref{u_ningamma+}, we have $q_n = 1$ for all $n = 0,\dots,M$.  
As $h^0=h$ by \eqref{expression_coefficients_0}, this implies with \eqref{second_order_filter} that for each fixed $n = 0,\dots,M$,
\begin{equation}\label{first_step}
    v_n = h_1^0 v_{n-1} + h_2^0 v_{n-2} + \dots + h_{k+1}^0 v_{n-(k+1)} + y - 1.
\end{equation}
The idea is to express $v_n$ in terms of the initial values $v_{-(k+1)},\dots,v_{-1}$ by repeatedly substituting the same relation for the terms $v_{n-1}, v_{n-2},\dots$.

In the first step, this means substituting $v_{n-1} = h_1^0 v_{n-2} + h_2^0 v_{n-3} + \dots + h_{k+1}^0 v_{n-(k+2)} + y - 1$ into \eqref{first_step},
which yields
\begin{align*}
    v_n 
    &= h_1^0\bigl(h_1^0 v_{n-2} + h_2^0 v_{n-3} + \dots + h_{k+1}^0 v_{n-(k+2)} + y - 1\bigr)
       + h_2^0 v_{n-2} + \dots + h_{k+1}^0 v_{n-(k+1)} + y - 1 \\
    &= \bigl((h_1^0)^2 + h_2^0\bigr) v_{n-2} 
       + \bigl(h_1^0 h_2^0 + h_3^0\bigr) v_{n-3} 
       + \dots 
       + h_1^0 h_{k+1}^0 v_{n-(k+2)} \\
    &\quad + (y-1)(1 + h_1^0) \\
    &= h_1^1 v_{n-2} + h_2^1 v_{n-3} + \dots + h_{k+1}^1 v_{n-(k+2)} 
       + (y-1)\sum_{i=0}^{1} h_1^{i-1},
\end{align*}
where the last equality follows from the recursive definitions of the coefficients $h_i^1$.

Iterating this procedure, at each step we substitute the expression for the “leading” term (e.g.\ $v_{n-2}$ in the line above) using the same type of identity, and we group together the coefficients of the state variables and of the constant term $(y-1)$ to obtain \eqref{characterization_of_trajectories}.  
\end{proof}

At this stage, the analysis of the trajectories can be divided into two main components:
First, the dependence on the trigger has been isolated in the
recursion. Second, rather than studying the dynamical system directly,
\eqref{second_order_filter}, we focus on analyzing the simpler system described by
items $i)$--$iii)$ of Proposition~\ref{evolution_highest_peak}, where the initial
conditions are fixed once the filter has been selected.

This reformulation enables us to study the trajectories by combining two
key elements: the evolution of the coefficients $(h_i^n)_{i,n}$ and the
representation of the trajectory elements provided by
\eqref{characterization_of_trajectories}, which makes the dependence on the trigger explicit.

To ensure the analysis applies to filters of arbitrary length—and to establish
stability results for bandlimited inputs—we approximate the evolution of the
coefficients $(h_i^n)_{i,n}$. This approach circumvents the need for a direct study of
their underlying dynamical system.

In the following, we present two lemmas for the coefficients $(h_i^n)_{i,n}$.  
The key observation, already apparent in Proposition~\ref{evolution_highest_peak}, is that these coefficients are interdependent.  
Our strategy is to express each coefficient in terms of $h_1^n$, and then to approximate the evolution of $h_1^n$.  
This, in turn, will automatically control all the remaining coefficients.

% \begin{lemma}\label{filter_coefficient_sum_up_to_1}
%     Let $h$ be a filter as in \eqref{minimal_filter}. Then $\sum_{i = 1}^{k+1} h_i^n = 1 \quad \text{for all } n \in \mathbb{N}.$
% \end{lemma}

% \begin{proof}
% We proceed by induction on $n$.

% For $n = 0$, by definition \eqref{expression_coefficients_0} we have $h^0 = h$, and hence
% \[
%     \sum_{i=1}^{k+1} h_i^0 \;=\; \frac{k+1}{k} - \frac{1}{k} \;=\; 1.
% \]

% For the induction step, assume that $\sum_{i=1}^{k+1} h_i^n = 1$
% for some $n > 0$.  
% Using the recursive definitions of the coefficients, we prove the statement for the level $n+1.$ Thus,
% \begin{gather*} \sum_{i = 1}^{k+1}h_i^{n+1} = \sum_{i = 1}^{k}(h_i^0 h_1^{n} + h_{i+1}^n) +h_{k+1}^0h_1^n \\ =h_1^n \sum_{i = 1}^{k}h_{i}^0 + \sum_{i = 2}^{k+1}h_i^n + h_{k+1}^{0}h_1^n \\
% = h_1^n\underbrace{\sum_{i = 1}^{k+1}h_i^0}_{=1} + \sum_{i = 2}^{k+1} h_i^n = \sum_{i = 1}^{k+1}h_i^n = 1. \end{gather*}
% where we have used the base case $\sum_{i=1}^{k+1} h_i^0 = 1$ and the induction hypothesis $\sum_{i=1}^{k+1} h_i^n = 1$ for the last equality.  
% \end{proof}

\begin{lemma}\label{lemma_evolution_h_i}
    Let $h$ be a filter as in \eqref{second_order_filter}. Then 
    \begin{itemize}
    \item[i)] for all $ n \in \mathbb{N}$ 
    \begin{equation}
        \sum_{i = 1}^{k+1} h_i^n = 1;
    \end{equation}
        \item[ii)] for $n\in \mathbb{N}$ and $2\le i \le k+1$ we have 
        \begin{equation}
            \left\{\begin{matrix}
                h_i^n = 0 & \text{if} \quad n+i \le k,\\\\
                h_i^n = -\frac{1}{k}h_1^{n+i - (k+2)} & \text{if} \quad n+i\ge k+1;
            \end{matrix}\right.
        \end{equation}
        \item[iii)] for all $1\le n<k$ 
        \begin{equation}
            h_1^n = \left(1+\frac1k\right)^{n+1};
        \end{equation}
        \item[iv)] for all $n\ge k$
        \begin{equation}
            h_1^{n} = \left(1+\frac1k\right)h_1^{n-1} -\frac{1}{k}h_1^{n-1-k}.
        \end{equation}
    \end{itemize}
\end{lemma}

\begin{proof}
    i) We proceed by induction on $n$.

    For $n = 0$, by definition \eqref{expression_coefficients_0} we have $h^0 = h$, and hence
    \[
        \sum_{i=1}^{k+1} h_i^0 \;=\; \frac{k+1}{k} - \frac{1}{k} \;=\; 1.
    \]
    For the induction step, assume that the property holds for $n\in \mathbb{N}^*,$ we prove the statement for the level $n+1.$ \\
 By $i),ii)$ and $iii)$ of Proposition \ref{lemma_evolution_h_i} it holds that for any $n\in \mathbb{N}$ the vector $(h_1^n,h_2^n,\dots, h_{k+1}^n)$ corresponds to the first row of $H^{n+1}.$ Therefore, $\sum_{i = 1}^{k+1}h_i^{n+1}$ can be obtained by considering the first element of the vector given by $H^{n+2}(1,\dots,1)^T.$ The property holds by observing that 
 \begin{gather*}
     H^{n+2}\left(\begin{matrix}
         1\\ 
         \vdots\\
         1
     \end{matrix}\right) = H \underbrace{H^{n+1} \left(\begin{matrix}
         1\\ 
         \vdots\\
         1
     \end{matrix}\right)}_{ = (1,\dots ,1)^T} = H\left(\begin{matrix}
         1\\ 
         \vdots\\
         1
     \end{matrix}\right) = \left(\begin{matrix}
         1\\ 
         \vdots\\
         1
     \end{matrix}\right),
 \end{gather*}
 where we used the induction step together with the structure of the matrix $H.$

    ii) If $n+i\le k,$ and hence $2\le i\le k-1,$ then $h_i = h_{i+1} = \dots = h_{n+i} = 0$  as $h$ is minimally supported. Consequently, one obtains from $ii)$ in Proposition \ref{evolution_highest_peak} that
    \begin{equation*}
        h_i^n = \underbrace{h_i}_{=0}h_1^{n-1} + h_{i+1}^{n-1} = \underbrace{h_{i+1}}_{=0}h_1^{n-2} +h_{i+2}^{n-2} = \dots = h_{n+i}^0 =h_{n+i}=0.
    \end{equation*}
    
    Similarly, if $n+i\ge k+1,$ then 
    \begin{equation*}
        h_i^n = \underbrace{h_i}_{=0}h_1^{n-1} + h_{i+1}^{n-1} = \underbrace{h_{i+1}}_{=0}h_1^{n-2} +h_{i+2}^{n-2} = \dots = h_{k+1}^{n+i - (k+1)} = -\frac{1}{k} h_{1}^{n+i - (k+2)}.
    \end{equation*}
    where we used iii) of Proposition \ref{evolution_highest_peak} in the last equality.\\\\
    iii) From ii) in Proposition \ref{evolution_highest_peak} and using ii) with the observation that $n-1\le k-2$ we have  
    \[
    h_1^n = h_1h_1^{n-1} + \underbrace{h_{2}^{n-1}}_{=0} = (h_1)^2h_1^{n-2} + h_1\underbrace{h_{2}^{n-2}}_{=0} =\dots = (h_1)^{n+1} = \left(1+\frac1k\right)^{n+1}.  
    \]
    iv) From ii) we have 
    \[
    h_1^n = \left(1+\frac1k\right)h_1^{n-1} + h_2^{n-1} = \left(1+\frac1k\right)h_1^{n-1} -\frac{1}{k}h_1^{n-1-k}.
    \]
\end{proof}

\begin{lemma}\label{bounded_finite_difference_of_h_1}
    Let $h$ be a filter as in \eqref{second_order_filter}, and set 
    \[
        \alpha := \frac{1}{k}\Bigl(1+\frac{1}{k}\Bigr),
        \qquad
        \beta := \frac{1}{k}\Bigl(1+\frac{1}{k}\Bigr)^{k-1}.
    \]
    Then, for the finite differences $\Delta h_1^n := h_1^n - h_1^{n-1}$ we have
    \[
        \alpha \;\le\; \Delta h_1^n \;\le\; \beta 
        \qquad \text{for all } n>1.
    \]
\end{lemma}

\begin{proof}
From (\textit{iv}) in Lemma~\ref{lemma_evolution_h_i} we have
\[
    \Delta h_1^n 
    = \frac{1}{k}\bigl(h_1^{n-1} - h_1^{n-k-1}\bigr)
    = \frac{1}{k}\sum_{i = 1}^{k}\Delta h_1^{n-i}
    \quad \text{for all } n>k.
\]
Thus, the sequence $(\Delta h_1^n)_{n>k}$ evolves as a moving average of $k$ consecutive terms.  
In particular, if we denote
\[
    \tilde{\alpha} := \min_{1 \le i \le k} \Delta h_1^i,
    \qquad
   \tilde{\beta}  := \max_{1 \le i \le k} \Delta h_1^i,
\]
then for all $n\ge k$ one has
\[
    \tilde{\alpha} \;\le\; \Delta h_1^n \;\le\;  \tilde{\beta}.
\]
We briefly show this by induction.
The property holds by construction at level $k$.
Assume now that it holds for some $n \ge k+1$, that is,
\[
    \tilde{\alpha} \le \Delta h_1^{n-i} \le \tilde{\beta},
    \qquad \text{for all } i = 1,\dots,k.
\]
Summing over $i=1,\dots,k$, we obtain
\[
    \tilde{\alpha}
    \le \sum_{i=1}^{k} \Delta h_1^{n-i}
    = \Delta h_1^{n+1}
    \le \tilde{\beta},
\]
which proves that the property also holds at level $n+1$.\\
It remains to compute these extrema explicitly; we will show that $\tilde{\alpha} = \alpha$ and  $\tilde{\beta} = \beta.$

From ii) in Lemma~\ref{lemma_evolution_h_i}, we know that
\[
    h_1^i = \Bigl(1+\frac{1}{k}\Bigr)^{i+1} 
    \qquad \text{for all } i = 0,\dots,k-1.
\]
Therefore, for $i = 1,\dots,k-1$
\begin{gather}\label{first_k-1}
    \Delta h_1^i
    = h_1^i - h_1^{i-1}
    = \Bigl(1+\frac{1}{k}\Bigr)^{i+1}
      - \Bigl(1+\frac{1}{k}\Bigr)^{i}
    = \frac{1}{k}\Bigl(1+\frac{1}{k}\Bigr)^i.
\end{gather}
The map $i \mapsto \bigl(1+\tfrac{1}{k}\bigr)^i$ is strictly increasing in $i>0$, so among $i=1,\dots,k-1$ the minimum is attained at $i=1$, the maximum at $i=k-1$:
\begin{gather}\label{Deltak}
    \Delta h_1^1 = \frac{1}{k}\Bigl(1+\frac{1}{k}\Bigr) = \alpha,
    \qquad
    \Delta h_1^{k-1} = \frac{1}{k}\Bigl(1+\frac{1}{k}\Bigr)^{k-1} = \beta.
\end{gather}

We now compute $\Delta h_1^k$ using (iii) in Lemma~\ref{lemma_evolution_h_i}.  
From
\[
    h_1^k = \frac{k+1}{k} h_1^{k-1} - \frac{1}{k}
          = \Bigl(1+\frac{1}{k}\Bigr)^{k+1} - \frac{1}{k},
\]
we obtain
\[
    \Delta h_1^k
    = h_1^k - h_1^{k-1}
    = \Bigl(1+\frac{1}{k}\Bigr)^{k+1} - \frac{1}{k}
      - \Bigl(1+\frac{1}{k}\Bigr)^{k}
    = \frac{1}{k}\biggl[\Bigl(1+\frac{1}{k}\Bigr)^{k} - 1\biggr].
\]

To compare $\alpha$ with $\Delta h_1^k$, note that
\[
    \Delta h_1^k > \alpha
    \quad \Longleftrightarrow \quad
    \Bigl(1 + \frac{1}{k}\Bigr)^{k} > 2 + \frac{1}{k},
\]
which holds for all $k\ge 3.$ Hence,
\[
    \tilde{\alpha} = \min_{1 \le i \le k} \Delta h_1^i = \alpha
    = \frac{1}{k}\Bigl(1+\frac{1}{k}\Bigr).
\]
 
It remains to compare $\beta$ with $\Delta h_1^k$. We have
\begin{align*}
    \beta - \Delta h_1^k
    &= \frac{1}{k}\Bigl(1+\frac{1}{k}\Bigr)^{k-1}
       - \frac{1}{k}\biggl[\Bigl(1+\frac{1}{k}\Bigr)^{k} - 1\biggr] \\
    &= \frac{1}{k}\biggl[
        \Bigl(1+\frac{1}{k}\Bigr)^{k-1} 
        - \Bigl(1+\frac{1}{k}\Bigr)^{k} + 1
      \biggr] \\
    &= \frac{1}{k}\Bigl[
        1 - \frac{1}{k}\Bigl(1+\frac{1}{k}\Bigr)^{k-1}
      \Bigr].
\end{align*}
For $k \ge 3$, we have
\[
    \Bigl(1+\frac{1}{k}\Bigr)^{k} < e < k,
\]
so in particular
\[
    \Bigl(1+\frac{1}{k}\Bigr)^{k-1} < k
    \quad \Longrightarrow \quad
    \frac{1}{k}\Bigl(1+\frac{1}{k}\Bigr)^{k-1} < 1.
\]
Thus $1 - \frac{1}{k}\bigl(1+\frac{1}{k}\bigr)^{k-1} > 0$, which implies
\[
    \beta - \Delta h_1^k > 0.
\]
Combining this with the monotonicity in $i$ for $1 \le i \le k-1$, we conclude that
\[
    \tilde\beta = \max_{1 \le i \le k} \Delta h_1^i = \beta
    = \frac{1}{k}\Bigl(1+\frac{1}{k}\Bigr)^{k-1}.
\]

This identifies the extremal values and proves the claimed bounds.
\end{proof}

\begin{remark}\label{linear_growth_of_h_1}
   As a consequence of Lemma \eqref{bounded_finite_difference_of_h_1}, as \(\alpha > 0\), the sequence \( (h_1^n)_{n \ge 0}\) is an increasing sequence that is bounded by two lines. Namely, for all \(n \ge 0\), we have
    \begin{equation}\label{bound_on_h_1^n}
    \alpha n + h_1^0 \le h_1^n \le \beta n + h_1^0.
    \end{equation}
\end{remark}

Our next goal is to prove upper bounds for the trajectory associated to triggers \eqref{characterization_of_trajectories} that depend only on the constant input and on the coefficients $h_1^n$, which can then be further simplified by plugging in \eqref{bound_on_h_1^n}.  
To this end, we introduce the notion of a \emph{critical trigger}.

\begin{definition}
    Let $h$ be a filter as in \eqref{minimal_filter} and consider a constant input sequence $y_n =y>1-\frac{2}{k}.$ Let $\bar{v} = (v_{-(k+1)}, \dots, v_{-1}) \in [-1,1]^{k+1}$ be a trigger of coverage $N(\bar{v},y)$. 
    %We define the coverage ({define earlier and eliminate this here}) of the trigger $N(\bar{v},y)$ as
    %\begin{equation}
    %    N(\bar{v},y):= \text{sup}\{m\in \mathbb{N} \ | \ v_j\in (1,+\infty)\ \ \text{for all} \ 0\le j\le m\},
    %\end{equation}
    %where the values $v_j$'s are obtained by the quantization of the contant input sequence with constant $y.$
    We say that $\bar{v}$ is \textit{critical} if for any other trigger $u\in[-1,1]^{k+1}$ we have that $N(u,y)\le N(\bar{v},y)$ and $$u_n \le v_n \ \text{for all } \ n = 0,\dots, N(u,y).$$ 
\end{definition}

\begin{figure}
    \centering
    \includegraphics[width=0.9\linewidth]{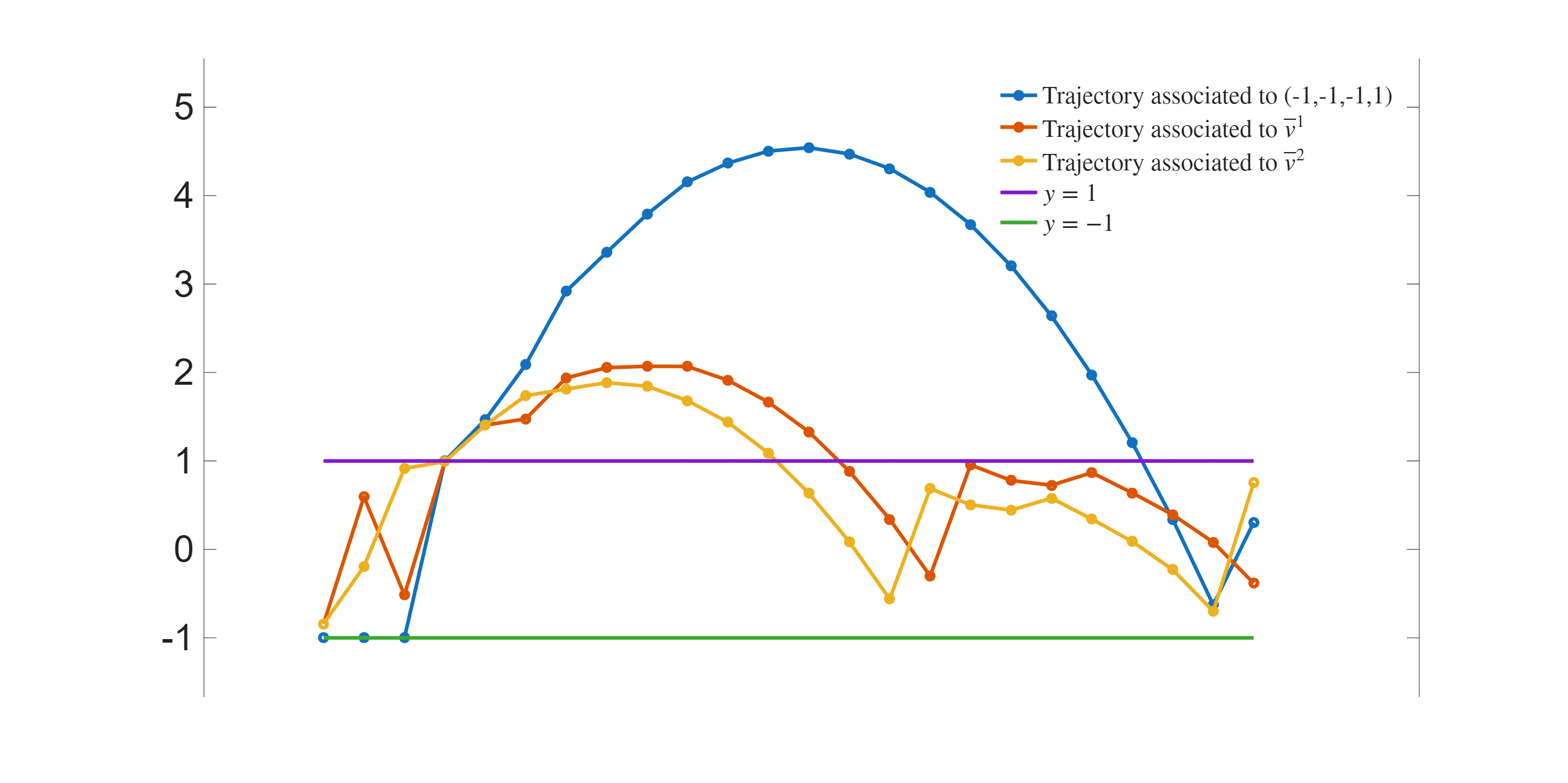}
    \caption{{\small We quantize the constant signal $y = 0.8$ with the filter $h = (0,-4/3,0,0,1/3)$ and three different  triggers $\bar{v} = (-1,-1,-1,1),$ $\bar{v}^1 = (0.99298,-0.5150,0.5938,-0.8430)$ and $\bar{v}^2 = (0.99298,0.9150,-0.1938,-0.8430)$ as initial conditions. The corresponding trajectories are shown. }}
    \label{plot_differece_triggers}
\end{figure}
To get an intuition for the notion of a critical trigger, we refer to Figure~\ref{plot_differece_triggers}.  
In this example, we consider the quantization of the constant input $y = 0.8$ and compare the trajectories generated by the critical trigger $\bar{v}$ and two other triggers, $\bar{v}_1$ and $\bar{v}_2$. For more details on how the critical trigger has been constructed, we refer the reader to Proposition \ref{Critical_Trigger}.  
We can read off from the graph that $N(\bar{v}_1,0.8) = 9$ and $N(\bar{v}_2,0.8) = 8$ and $N(\bar{v},y) =17$, so indeed the critical trigger has the greatest coverage. One also notes that along the trajectory, the state variable iterates corresponding to the critical trigger are maximal among the three cases. However, it is generally not the case that larger state variables correspond to a greater coverage. Indeed, although $\bar{v}_1$ has greater coverage than $\bar{v}_2$, the second state variable along its trajectory is smaller than the second state variable in the trajectory generated by $\bar{v}_2$.

It follows directly from the definition that an upper bound for the trajectory associated with a critical trigger provides a uniform bound for all other trajectories. Hence, a key step for our analysis will be to establish that a critical trigger always exists, see Proposition \ref{Critical_Trigger} below.

As a consequence, the trajectory generated by the critical trigger can serve as an envelope for the trajectories generated by all other triggers, and the analysis of all configurations can be reduced
to that single extremal case.

To prove the existence of a critical trigger, we need the following lemma.

% Our interest lies in understanding how trajectories that are initially bounded
% and eventually return to the interval $[-1,1]$ influence the dynamics in this
% region, where the system ceases to be linear.
% In particular, we seek a bound on the constant input
% $y > 1 - \tfrac{2}{k}$ ensuring that, once the state variables enter
% $(1,+\infty)$, they cannot subsequently cross into $(-\infty,-1)$.

% This is precisely the scenario we aim to rule out, as it would lead to instability
% mechanisms of the type illustrated in Figure~\ref{Critical_scenarios}.
% As will be shown later, although this analysis is carried out for constant input sequences, it is still fundamental for understanding the quantization of smooth functions, as it provides information on the maximal variation a function can exhibit without inducing the oscillatory behavior illustrated in Figure~\ref{fig:duefigure}.  

% In the following, we present some properties that allow us to identify the critical trigger and the trajectory it generates.

\begin{lemma}\label{sign_of_the_filter}
    Let \( h\) be a filter as in \eqref{minimal_filter}. For all \( n \ge 0 \), we have:
\[
\left\{
\begin{array}{cc}
 \ \operatorname{sign}(h_i^n) = -1 & \text{for } \ i = \operatorname{max}\{2,k+1 -n\},\dots,k+1, \\
\text{sign}(h_1^n) = 1 & 
\end{array}
\right.
\]
\end{lemma}
\begin{proof}
    From Remark \ref{linear_growth_of_h_1} it follows that $h_1^n \ge h_1^0 = \left(1+\frac{1}{k}\right)>0$, so we have $\text{sign}(h_1^n) = 1$ for all $n\ge 0.$ Moreover, from Lemma \ref{lemma_evolution_h_i} we have that for all $ i = \text{max}\{0,k+1 -n\},\dots,k+1,$  \[ h_i^n = -\frac{1}{k}\underbrace{h_1^{n+i - (k+2)}}_{>0}< 0,\] thus yielding the statement.
\end{proof}

\begin{prop}\label{Critical_Trigger}
    Let $h$ be a filter as in \eqref{minimal_filter} and consider a constant input sequence with $ 1-\frac{2}{k}<y<1.$ 
    Then $\bar{v} = (\underbrace{-1,\dots,-1}_{k},1)$ is the unique critical trigger and the generated trajectory is
    \begin{equation}\label{maximal_trajectory}
        v_n= \sum_{i = 1}^{k+1}\left|h_{i}^n\right| + (y-1)\sum_{i = 0}^n h_1^{i-1} ,
    \end{equation}
    
    for all $n = 0,\dots, N(\bar{v},y).$
\end{prop}
\begin{proof}\label{parabola}
    From Proposition \ref{evolution_highest_peak}, we know an explicit formula for the state sequence $v_n$ during the trajectory. For this proof, we extend this formula beyond the trajectory defining for all $n \ge 0$ and all triggering configurations $w\in [-1,1]^{k+1}$
   \begin{equation}\label{evolution_csi_n}
        \tilde{v}_n(w) := \underbrace{h_1^nw_{k+1} + h_2^nw_{k} + \dots + h_{k+1}^n w_1}_{(a)} + \underbrace{(y-1)\sum_{i = 0}^{n}h_1^{i-1}}_{(b)}.
   \end{equation}
   Observe that $(b)$ does not depend on the triggering configuration, therefore, to maximize $(a)$, and thus $\tilde{v}_n$, we need to ensure that $w_\ell = \operatorname{sign}(h_{k+2-\ell}^n)$ for those $\ell$, where $h_{k+2-\ell}^n \neq 0$; when $h_{k+2-\ell}^n = 0,$ the choice of $w_{\ell}\in [-1,1]$ is arbitrary.
  From Lemma \ref{sign_of_the_filter} we know that for all $n\ge 0,$ $h_1^n$ is positive, and for all $n\ge k-1$, $h_i^n$ is negative for all $i=2,\dots,k+1.$ Consequently, the unique triggering configuration that maximizes \eqref{evolution_csi_n} for $n \ge k-1$ is \[\bar{v} = (\text{sign}(h_{k+1}^n),\text{sign}(h_k^n),\dots, \text{sign}(h_1^n)) = (-1,-1,\dots,-1,1).\] 
   This configuration is also a (no longer unique) maximizer for $n <k-1,$ as it follows from Lemma \ref{lemma_evolution_h_i} and \ref{sign_of_the_filter} that $h_i^n$ is non-positive for all $i = 2,\dots,k+1$ also in these cases.

   To conclude that $\bar{v}$ is a critical trigger, we observe that the coverage of a trigger $w$ is the minimal $N>0$ such that $\tilde{v}_N(w)\le 1$. Consequently, we have for all $n < N(w,y)$ that $1<\tilde{v}_n(w)\le \tilde{v}_n(\bar{v})$ and hence $N(\bar{u},y)\le N(\bar{v},y).$ 
   % Thus, for all $n\ge 0$ we have
   % \[\tilde{v}_n(\bar{v}) = \sum_{i = 1}^{k+1}\left|h_{i}^n\right| + (y-1)\sum_{i = 0}^n h_1^{i-1}.\]
   
   % Observe that $\bar{v}$ is the critical trigger. Indeed, by construction, for any other trigger $\bar{u}$ we have that $u_n\le v_n$ for any $n\ge 0,$ which also implies that $N(\bar{u},y)\le N(\bar{v},y).$
   Together with the aforementioned maximization property, this proves that $\bar{v}$ is the unique critical trigger; the trajectory formula \eqref{maximal_trajectory} directly follows by plugging in $\bar{v}$.
\end{proof}

Combining the formula for the trajectory associated with the critical trigger with the approximation given by Lemma \ref{bounded_finite_difference_of_h_1}, we obtain a parabolic upper bound as given in the following proposition. By definition of the critical trigger, this parabola will provide upper bounds for the trajectories generated by any triggers. 
% At this point, for a fixed value $1 - \tfrac{2}{k} < y < 1$, we aim to establish two key facts.  
% First, we show that the trajectory associated with the critical trigger is bounded from above; since this trajectory dominates all others, this will imply that every trajectory generated by any trigger is bounded.  
% Second, we prove that such trajectory can be bounded from above by a parabola depending only on the constant input $y$ and on the length $k+1$ of the filter.
% We prove this in the following proposition.

\begin{prop}\label{parabola_behavior}
    Let $(v_n)_{n =0,\dots, M}\subset (1,+\infty)$ be a trajectory generated by a generalized positive triggering sequence for a constant input sequence with $1-\frac{2}{k}<y<1.$ Then the state trajectories satisfy the parabolic bound 
   {\small \[v_n \le \frac{y-1}{2k}\left(1+\frac{1}{k}\right)n^2 + \left[(y-1)\left(1+\frac{1}{k}\right)\left(1-\frac{1}{2k}\right) + \frac{2}{k}\left(1+\frac{1}{k}\right)^{k-1}\right]n + \frac{2}{k} + y \quad \text{for all } \quad n=0\dots,M .\]} 
   In particular, all the state variables in the trajectories are bounded by the maximum of this parabola, which is given by 
   \begin{equation}
       M(y,k) = y + \frac{2}{k} -\frac{\left[\gamma(y-1) + \mu\right]^2}{2\alpha(y-1)},
   \end{equation}
   where $\gamma = \left(1+\frac{1}{k}\right)\left(1-\frac{1}{2k}\right)$, $\mu = \frac{2}{k}\left(1+\frac{1}{k}\right)^{k-1}$ and $\alpha = \frac{1}{k}(1+\frac1k).$
\end{prop}
\begin{proof}
Let $(\xi_n)_{n=0,\dots,M}$ denote the first $M+1$ elements of the trajectory
generated by the critical trigger, as defined in
Proposition~\ref{Critical_Trigger}.
Observe that $v_n \le \xi_n$ for all $n  = 0,\dots, M$. Indeed, as $(v_{-(k+1)},\dots, v_{-1})$ generates a trajectory in $(1,+\infty)$ and $v_2,\dots, v_{(-(k+1))}\in [-1,+\infty)$, from Lemma \ref{sign_of_the_filter} one has for all $n = 0,\dots,M$
\begin{equation}\label{bound_any_conf}
    v_n = \underbrace{h_1^n v_{-1}}_{\le |h_1^n|} +\dots +\underbrace{v_{-(k+1)}h_{k+1}^n}_{\le |h_{k+1}^n|} + (y-1)\sum_{i = 0}^{n}h_1^{i-1} \le \xi_n = \sum_{i=1}^{k+1} \lvert h_i^n\rvert
      + (y-1)\sum_{i=0}^{n} h_1^{\,i-1}.
\end{equation}

Observe that, using Lemma~\eqref{sign_of_the_filter} and $i)$ of Proposition~\ref{lemma_evolution_h_i} one has
\[
    \sum_{i=1}^{k+1} \lvert h_i^n\rvert
    = h_1^n - \sum_{i=2}^{k+1} h_i^n
    = 2h_1^n - 1,
    \qquad \text{for all } n \ge 0
\]
and thus, the right-hand side of \eqref{bound_any_conf} simplifies to
\begin{equation}\label{simp_max_traj}
    v_n \le 2h_1^n - 1
    + (y-1)\sum_{i=0}^{n} h_1^{\,i-1}.
\end{equation}

From Lemma~\eqref{bounded_finite_difference_of_h_1}, we know that for all
$n \ge 0$,
\begin{equation}\label{bounds_h1}
    \alpha n + h_1^0
    \le h_1^n
    \le \beta n + h_1^0,
\end{equation}
where
\[
    \alpha = \frac{1}{k}\left(1+\frac{1}{k}\right),
    \qquad
    \beta = \frac{1}{k}\left(1+\frac{1}{k}\right)^{k-1},
    \qquad
    h_1^0 = 1+\frac{1}{k}.
\]
Plugging this into the right-hand side of \eqref{bound_any_conf} and using that $y-1<0,$ we obtain that 
% Hence, it suffices to bound the coefficients $h_1^i$ in
% \eqref{simp_max_traj}. 
% Since $y-1<0$, we observe that
% \begin{align*}
%     (y-1)\sum_{i=1}^{n} h_1^{\,i-1}
%     &\le (y-1)\sum_{i=1}^{n} \bigl(\alpha(i-1) + h_1^0\bigr) \\
%     &= (y-1)h_1^0\, n
%        + \frac{(y-1)\alpha}{2}\, n(n-1).
% \end{align*}
% Using the upper bound in \eqref{bounds_h1}, we obtain
\begin{align*}
    v_n
    &\le 2(\beta n + h_1^0) - 1
      + (y-1)
      + (y-1)\sum_{i=1}^{n} \bigl(\alpha(i-1) + h_1^0\bigr)\\
      &\le \frac{(y-1)\alpha}{2}\, n^2
    + \Bigl[(y-1)\Bigl(h_1^0 - \frac{\alpha}{2}\Bigr)
      + 2\beta\Bigr] n
    + 2h_1^0 - 2 + y.
\end{align*}
Substituting the explicit expressions for $\alpha$, $\beta$, and
$h_1^0$ into the inequality above gives the desired result.
\end{proof}

% Finally, thanks to Proposition~\eqref{parabola_behavior}, we can now bound from above any trajectory in $(1,+\infty)$ generated by a trigger, and this bound is independent of the particular trigger.

% \begin{remark}\label{Important_remark}
% Note that the bound for the trajectories established in Proposition~\ref{parabola_behavior} also applies to any configuration 
% $\bar{v}$ such that $\bar{v}_1 \in [-1,1]$ and 
% $(\bar{v}_2,\dots,\bar{v}_{k+1}) \in (-1,+\infty)^{k}$, which generates a trajectory in $[1,+\infty)$.  
% This again follows from the fact that the coefficients $h_i^n$ are negative for all $i = 2,\dots,k+1$, while $h_1^n \ge 0$.  
% Therefore, in~\eqref{evolution_csi_n}, each term $\bar{v}_{k+1-i} h_i^n$ would still be bounded above by $|h_i^n|$.  
% This observation will play a crucial role in establishing the boundedness of $v$ for the quantization of both constant case and non constant case.
% \end{remark}

\subsection{Analysis on the stability for constant input sequences}

In the following section, we introduce the key definitions and propositions that will allow us to state the main theorems concerning the stability of the proposed schemes. A crucial step will be to consider the case of constant signals (that are also of independent interest, see for example \cite{gunturk2004refined}). As we will see (cf. Theorem \ref{On the stability for constant inputs} below) the following condition induces stability for this case.   

\begin{definition}\label{stabilizing_input}
    Let $1-\frac{2}{k}<y<1$ and let 
    \[
    M(y,k) = y + \frac{2}{k} -\frac{\left[\gamma(y-1) + \mu\right]^2}{2\alpha(y-1)},
    \] be the maximum of the parabola as in Proposition \ref{parabola_behavior}
    where $\gamma = \left(1+\frac{1}{k}\right)\left(1-\frac{1}{2k}\right)$, $\mu = \frac{2}{k}\left(1+\frac{1}{k}\right)^{k-1}$ and $\alpha = \frac{1}{k}(1+\frac1k).$
    We say that $y$ is a \textit{stabilizing constant input} if 
    \begin{equation}\label{stabilizing_equation}
        -\frac{k+1}{k} -    \frac{M(y,k)}{k} +y+1\ge -1,
    \end{equation}    
\end{definition}

The reason that condition \eqref{stabilizing_equation} ensures stability is that it prevents the state variable from taking values $v_n<-1$ as its left-hand side bounds the state sequence from below, see the proof of Theorem \ref{On the stability for constant inputs} for details.

% whenever the state variable.
% First, we observe that the quantity $M(y,k)$ is obtained by computing the maximum
% of the parabola given in Proposition~\ref{parabola}.

% The main mechanism that may lead to instability is that the peak of the triggered
% trajectory becomes sufficiently large. In this case, after the state variable
% re-enters the interval $[-1,1]$, one may encounter a configuration in which
% $v_{n-1}\in[-1,1]$ while $v_{n-(k+1)}$ coincides with the trajectory peak.
% Consequently, there is a risk that
% \[
%     v_n
%     = \frac{k+1}{k} v_{n-1}
%       - \frac{1}{k} v_{n-(k+1)}
%       + y - q_n
%     < -1,
% \]
% and repeated iterations of the scheme may therefore result in instability.

% In this framework, we adopt a worst-case analysis.
% From Lemma~\ref{u_ningamma+}, we know that if $q_n = 1$, then $v_n \ge -1$.
% Hence, it suffices to consider the case $q_n = -1$.
% Under this assumption, we further consider the configuration that makes $v_n$
% as negative as possible, namely
% \[
%     v_{n-1} = -1,
%     \qquad
%     v_{n-(k+1)} = M(y,k).
% \]
% This choice leads directly to equation~\eqref{stabilizing_equation}.

% Our goal is therefore to determine how much one can exceed the classical
% criterion~\eqref{classical_criterion}, and therefore increasing the admissible $y$, while still ensuring that $y$ is a stabilizing constant input. 

To later generalize the stabilizing property to non-constant inputs, we need robustness with respect to perturbations, which follows from a strict inequality in \eqref{stabilizing_equation} and is hence quantified by the gap function as given in the following definition.

% A key aspect in the stability analysis, both in the constant and in the
% non-constant input case, is the distance between the signal amplitude and the
% cut-off threshold introduced in Definition~\ref{stabilizing_input}.
% This distance quantifies the robustness margin of the scheme and motivates the
% following definition.

\begin{definition}\label{Definition: gap function}
    For every $1-\frac{2}{k}<y<1,$ let \[
    \tilde{y}(y,k):= \frac{1}{k}-1 + \frac{M(y,k)}{k}.
    \]
    Then the \text{gap function} is given by
    $g(y,k):= y-\tilde{y}(y,k).$
\end{definition}
% The \emph{gap function} is designed to identify the conditions under which stable quantization can be done while violating the classical stability criteria. 
% We will discuss this definition in more detail in the next section.
% In the following, we analyze some key properties of the gap function.
\begin{prop}\label{Properties of the gap function}[Properties of the gap function] 
Let $k\ge 3$ and $1-\frac{2}{k}<y<1$.  
Then the gap function
\[
g(y,k):=y-\tilde y(y,k)=y+1-\frac{1}{k}-\frac{M(y,k)}{k}
\]
satisfies the following properties.
\begin{itemize}
    \item[i)] $g(\cdot,k)$ is a continuous concave and strictly decreasing function on $\bigl[1-\frac{2}{k},\,1\bigr)$ for any $k\ge 3$. Moreover, it holds that $g\left(1-\frac{2}{k},k\right) > 1/2$ and $\operatorname{lim}_{y\to 1^-}g(y,k) = -\infty$.
    \item[ii)] There exists a unique $y(k)\in \left(1-\frac2k,1\right)$ such that $g(y(k)) = 0.$
    \item[iii)] It holds that $y(k)>1-\frac{e^2}{(k+1)^2}.$
\end{itemize}
\end{prop}

\begin{proof}
$i)$ For simplicity, we will fix $k$ and omit it in the notation of $g$ and $M$.
Since
\begin{gather*}
    M(y) = y+\frac{2}{k} -\frac{(\gamma(y-1) +\mu)^2}{2\alpha (y-1)} =  y+\frac{2}{k} -\frac{\gamma\mu}{\alpha} + \frac{\gamma^2}{2\alpha}(1-y) + \frac{\mu^2}{2\alpha(1-y)},
\end{gather*}
we have 
\begin{gather*}
    g(y)  
    % y +1 - \frac{1}{k} - \frac{y}{k} -\frac{2}{k^2} +\frac{2\beta\mu}{2k\alpha} -\frac{\beta^2}{2k\alpha}(1-y) - \frac{\mu^2}{2k\alpha(1-y)} \\
    = 1 - \frac{1}{k} -\frac{2}{k^2} +\frac{\gamma\mu}{k\alpha} +y\left(1-\frac{1}{k}\right) -\frac{\gamma^2}{2k\alpha}(1-y) - \frac{\mu^2}{2k\alpha(1-y)}\\
    = A_0  + yA_1 - A_2(1-y) - \frac{A_3}{1-y},
\end{gather*}

with 
\begin{gather*}
    A_0 = 1 - \frac{1}{k} -\frac{2}{k^2} +\frac{\gamma\mu}{k\alpha},\\
    A_1 = \left(1-\frac{1}{k}\right),\\
    A_2 = \frac{\gamma^2}{2k\alpha},\\
    A_3 = \frac{\mu^2}{2k\alpha}.
\end{gather*}
Therefore, it follows that $g$ is continuous in $(0,1)$ and since $A_3 \ge 0$ it is immediate to see that $\lim_{y\to 1^-}g(y) = -\infty$.\\
Now we prove that $g(y)$ is a strictly decreasing function by showing that its derivative, given by
\begin{gather*}
    g'(y) = A_1 + A_2 - \frac{A_3}{(1-y)^2},
\end{gather*}
is negative for all $y\in \left(1-\frac2k,1\right)$. Since $A_3 \ge 0$ we have that $g'(y)$ is a decreasing function in $(0,1),$ thus, we only need to verify the negativity for $g'\left(1-\frac2k\right)$

For $y_0=1-\frac{2}{k}$, we have $1-y_0=\frac{2}{k}$ and thus
\[
g'(y_0)=A_1+A_2-\frac{k^2}{4}A_3.
\]
Inserting the value of $\alpha$, $\gamma$ and $\mu$ into $A_2$ and $A_3$ yields
% Since $\alpha=\frac{k+1}{k^2}$ and $\beta = ...$, it holds that $\frac{1}{2k\alpha}=\frac{k}{2(k+1)}$ and hence
\[
A_2=\frac{k}{2(k+1)}\gamma^2
=\frac{k}{2(k+1)}\cdot\frac{(k+1)^2(2k-1)^2}{4k^4}
=\frac{(k+1)(2k-1)^2}{8k^3},
\]
and
% , using that $\mu^2=\frac{4}{k^2}\bigl(1+\frac1k\bigr)^{2k-2}$,
\[
\frac{k^2}{4}A_3
=\frac{k^2}{4}\cdot\frac{k}{2(k+1)}\mu^2
=\frac{k}{2(k+1)}\Bigl(1+\frac1k\Bigr)^{2k-2}.
\]
Therefore
\begin{equation}\label{eq:gprime-y0}
g'\!\left(\frac{k-2}{k}\right)
=
\left(1-\frac1k\right)
+\frac{(k+1)(2k-1)^2}{8k^3}
-\frac{k}{2(k+1)}\Bigl(1+\frac1k\Bigr)^{2k-2}.
\end{equation}

\smallskip
\noindent\textbf{Step 1: an upper bound for the positive terms.}
For $k\ge 4$ we have
\[
\frac{(k+1)(2k-1)^2}{8k^3}
<\frac{(k+1)4k^2}{8k^3}=\frac{k+1}{2k}\le \frac58.
\]
Hence, for all $k\ge 4$,
\begin{equation}\label{eq:pos}
\left(1-\frac1k\right)+\frac{(k+1)(2k-1)^2}{8k^3}<\frac{11}{8}.
\end{equation}

\smallskip
\noindent\textbf{Step 2: a lower bound for the negative term.}
As $(1+1/k)^k$ is increasing in $k\ge 4$, it holds that
$\left(1+\frac1k\right)^k \ge \left(1+\frac14\right)^4=\left(\frac54\right)^4$, and thus
\[
\left(1+\frac1k\right)^{2k-2}
=\left[\left(1+\frac1k\right)^k\right]^{\,2-\frac{2}{k}}
\ge \left[\left(\frac54\right)^4\right]^{3/2}=\left(\frac54\right)^6 .
\]
Moreover, as $\frac{k}{2(k+1)}\ge \frac{2}{5}$ for $k\ge 4$, one obtains that 
\begin{equation}\label{eq:neg}
\frac{k}{2(k+1)}\left(1+\frac1k\right)^{2k-2}
\ge \frac{2}{5}\left(\frac54\right)^6.
\end{equation}
Observing that $\frac{2}{5}\left(\frac54\right)^6 \ge \frac{11}{8}$ we have that for every $k\ge 4$, the negative term in \eqref{eq:gprime-y0} dominates the positive ones, thus
\[
g'\!\left(1-\frac{2}{k}\right)<0.
\]

For the case $k = 3$, it is possible to verify directly that the inequality holds.\\\\
It remains to show that $g\left(1-\frac2k\right)>\frac12.$
Recall that
\[
g(y)=\Bigl(1-\frac1k\Bigr)y+1-\frac1k-\frac{2}{k^2}
+\frac{(\gamma(y-1)+\mu)^2}{2\alpha k (y-1)}.
\]
Let $y_0=\frac{k-2}{k}=1-\frac{2}{k}$. Then plugging in $y_0-1=-\frac{2}{k}$ in the above formula yields
\[
g(y_0)= 2-\frac{4}{k} -\frac{1}{k+1}\left(\left(1+\frac{1}{k}\right)^{k-1} - \left(1+\frac{1}{k}\right)\left(1-\frac{1}{2k}\right)\right)^2,
\]
Observe that $\left(\left(1+\frac{1}{k}\right)^{k-1} - \left(1+\frac{1}{k}\right)\left(1-\frac{1}{2k}\right)\right)^2\le(e-1)^2$ and thus
\[
g(y_0) \ge 2-\frac{4}{k} -\frac{(e-1)^2}{k+1}.
\]
The right hand-side is an increasing function and for $k = 5$ it holds that $g(y_0)\ge 1/2.$ The statement can be verified directly for the values $k = 3,4$.\\\\
$ii)$ Direct consequence of $i).$\\\\
$iii)$ It is enough to show that \(g\left(1-\frac{e^2}{(k+1)^2}\right)>0\), since \(g(y)\to-\infty\) as
\(y\to1^{-}\); hence, the right zero must lie to the right of any point where
\(g\) is still positive.
Let $t\in (0,1)$. By reformulating the gap function and plugging in the value $1-e^2/(k+1)^2$ we obtain
\begin{align*}
g\left(1-\frac{e^2}{(k+1)^2}\right)
=
\frac{2(k^2-k-1)}{k^2}
+
\frac{(2k-1)\left(1+\frac{1}{k}\right)^{k-1}}{k^2}
\\-
\left(
\frac{k-1}{k}
+
\frac{(k+1)(2k-1)^2}{8k^3}
\right)\frac{e^2}{(k+1)^2}
-
\frac{2(k+1)\left(1+\frac{1}{k}\right)^{2(k-1)}}{ke^2}.
\end{align*}
By using the fact that $2k/(k+1)<(1+1/k)^{k-1}<ek/(k+1)$, the right-hand side of the above equality can be lower-bounded as follows 
\begin{align*}
g\left(1-\frac{e^2}{(k+1)^2}\right)>\frac{2(k^2-k-1)}{k^2}
+
\frac{2(2k-1)}{k(k+1)}
-
\frac{8(k-1)}{k(k+1)^2}
-
\frac{(2k-1)^2}{k^3(k+1)}
-
\frac{2k}{k+1}\\
= \frac{4k^4-14k^3+k-1}{k^3(k+1)^2}.
\end{align*}
For $k\ge4$, one has $4k^4-14k^3+k-1 = k^3(4k-14) +k -1 >0$, thus the value $(1-e^2/(k+1)^2)$ is a minorant for the unique zero of $g$. 

The positivity for the case $k = 3$ can be verified directly by plugging in the value $1-e^2/16$ into the gap function.
This concludes the proof.

\end{proof}

\begin{remark}\label{another_trigger}
% From the previous two proposition, observing that $g\to -\infty$ as $y\to 1^-,$ we conclude that there exists, for each $k,$ a maximal stabilizing input $y(k)$, such that $g(y(k)) = 0$. 

From Proposition \ref{Properties of the gap function} we know from that \( g(x) \geq 0 \) for all \( x \in \left[1-\frac{2}{k}, y(k)\right] \). Consequently, all values within the range \(1- \frac{2}{k} < y \leq y(k) \) are stabilizing constant inputs as they satisfy
\begin{equation*}
    -\frac{k+1}{k} -    \frac{M(y,k)}{k} +y+1 \ge -\frac{k+1}{k} -    \frac{M(y,k)}{k} +\tilde{y}(y,k)+1 = -1 ,
\end{equation*}
where the inequality follows from the fact that the gap function is positive at $y$.
Furthermore, there is an important consequence that will be helpful for the bandlimited case. Since \( g(x) > 0 \) for all \( x \in \left[1-\frac{2}{k}, y(k)\right) \), if a trigger pushes some \( v_n \) into \( (1,+\infty) \), then during the decaying phase of the trajectory, a small decrease in the function value is permissible as long as it does not fall below \( \tilde{y}(y,k) \). This  will ensures that at least one point in the interval \([-1, 1]\) can still be reset. This key observation will help us prove that smooth functions can be quantized stably with a minimum level of oversampling.
\end{remark}

With these tools in hand, we are ready to state the theorem that ensures stability beyond the criterion \ref{classical_criterion} for $k\ge 3$.
\begin{theorem}[Stability for constant inputs]\label{On the stability for constant inputs}
    Let $h$ be a filter as in \eqref{minimal_filter}, with $k\ge 3$ and $0< y \le  1-\frac{e^2}{(k+1)^2}$.\\ Then the $\Sigma\Delta$ scheme with filter $h$, constant input $y$ and initial conditions $\tilde{v}\in[-1,1]^{k+1}$ is stable. Moreover, \[
    -1 \le v_n \le
    \begin{cases}
       M(y,k) & \textit{if} \ y>1-\frac2k \\
       1 & \text{if} \ y\le 1-\frac{2}{k}
       
    \end{cases} \quad \text{for all} \quad n\in \mathbb{N}_0,
    \]
    with  \[
    M(y,k) = y + \frac{2}{k} -\frac{\left[\gamma(y-1) + \mu\right]^2}{2\alpha(y-1)},
    \]
     where $\gamma = \left(1+\frac{1}{k}\right)\left(1-\frac{1}{2k}\right)$, $\mu = \frac{2}{k}\left(1+\frac{1}{k}\right)^{k-1}$ and $\alpha = \frac{1}{k}(1+\frac1k).$
\end{theorem}

\begin{proof}
Since for $y\le1-\frac{2}{k}$ the criterion \eqref{classical_criterion} is satisfied, it follows that $v_n\in[-1,1]$ for all $n\in \mathbb{N}_0$. It remains to show the case $y>1-\frac{e^2}{(k+1)^2}$.

First of all, recall from Remark \ref{another_trigger} that all the $y$'s that are lower than the zero $y(k)$ of the gap function are stabilizing constant inputs, and that by $iii)$ of Proposition \ref{Properties of the gap function} $ 1-\frac{e^2}{(k+1)^2}< y(k)$. Consequently, any $0<y\le 1-\frac{e^2}{(k+1)^2}$ is a stabilizing constant input and 
satisfy equation \eqref{stabilizing_equation}, i.e.,
\begin{equation}\label{particular_case_stab}
    -\frac{k+1}{k} - \frac{M(y,k)}{k} + y+1\ge-1.
\end{equation}

Now we prove by contradiction that 
\begin{equation}\label{Stab_input_hold}
   v_n\ge-1 \quad \text{for all} \ n\in \mathbb{N}_0.
\end{equation}
Assume that there exists an index such that Property \eqref{Stab_input_hold} does not hold, and let  $N\in \mathbb{N}_0$ be the first index such that $v_N<-1,$ and therefore 
\begin{equation}\label{Fact_about_previous_state_variables}
    v_{n}\ge -1 \quad \text{for all} \ n = 0,\dots,N-1.
\end{equation}
Furthermore, if $v_{N-(k+1)}>1$, then by \eqref{Fact_about_previous_state_variables} and the fact that initialization state vectpr $\tilde{v}$ is in $[-1,1]^{k+1}$ it must belong to some positive trajectory generated by a generalized positive trigger and thus, by Proposition \ref{parabola_behavior}, it holds that $v_n\le M(y,k).$ Consequently, since $v_N <-1$ and hence $q_{N} = -1$, one has that
\begin{gather*}
    v_N = \frac{k+1}{k}v_{N-1} - \frac{1}{k}v_{N-(k+1)} +y +1\ge -\frac{k+1}{k} - \frac{M(y,k)}{k} + y+1\ge-1,
\end{gather*}
where the last inequality follows from \eqref{particular_case_stab}. This contradicts the assumption.\\
The upper bound follows from the fact that by Property \eqref{Stab_input_hold} any trajectory is upper bounded by $M(y,k)$. This concludes the proof.

\end{proof}

\begin{remark}
    Note that the theorem can be slightly generalized by considering $y(k)$ implicitly defined via $g(y(k)) = 0$ as an upper bound for $y$ and $$M_s(y,k) = \text{max}\{2h_1^n -1 + (y-1)\sum_{i = 0}^n h_1^{i-1} \ : \ n = 0,\dots, N(\bar{v},y)\}$$
    instead of $M,$ where $\bar{v}$ is the positive critical trigger. One obtains that for all $0\le y\le y(k)$ it holds that 
    \[
    -1 < v_n < M_s(y,k) \quad \text{for all} \quad n\in \mathbb{N}.
    \]
However, due to the implicit nature of these conditions, they hide the order of the $k$ dependence, which is why we mainly work with the sub-optimal conditions of Theorem \ref{On the stability for constant inputs}. 
\end{remark}

%     Let $h = \left(\frac{k+1}{k},\dots,-\frac{1}{k}\right),$ with $k\ge 3,$ and $0< y \le y(k)$, where $y(k)$ is the solution of $g(y(k)) = 0.$\\ The $\Sigma\Delta$ scheme with feedback filter $h$, constant input $y$ and initial conditions $\bar{v}\in[-1,1]^{k+1}$ is stable. Moreover,
%     where 
%     and $\bar{v} = (1,\underbrace{-1,\dots,-1}_{k}).$

% \begin{proof}
%     The proof is the same as in Theorem \eqref{On the stability for constant inputs}. Observe that $M_s(y,k)\le M(y,k)<+\infty.$
% \end{proof}

\subsection{Stability for bandlimited functions}
If a non-constant function is to be quantized, some crucial cancellation in the above argument may fail unless an oversampling rate high enough to exploit smoothness properties. This is illustrated by the following example (inspired by \cite{yilmaz2002stability}), which shows that for a non-smooth input, schemes that are provably stable for constant inputs can become unstable.

% Here, the bandlimitedness of the function  In this scenario, the gap function plays an additional role: it allows us to determine the minimum sampling rate required for a bandlimited function to be quantized in a stable way.

% To clarify its main intuition and use, we present the following example.

\begin{example}\label{example}
Suppose we aim to quantize a step function with samples $$y_n = \begin{cases}
    0.7 & \text{if} \ n\le 6\\
    -0.7 & \text{if} \ n>7,
\end{cases}$$ using the filter with $k = 3,$ i.e.,
\[
h = \left(0,-\tfrac{4}{3}, 0, 0,\tfrac{1}{3}\right)
\] 
and the initial condition $\bar{v} = (-1,-1,\dots,-1,1),$ which is exactly the critical trigger identified for constant inputs. A direct calculation yields that $v_0>1,$ so $\bar{v}$ is a trigger and generates a trajectory in $(1,+\infty).$ As one can see in Figure \ref{fig:fig1}, exactly at the peak of the trajectory, the input changes sign, which leads to an accelerated decrease and eventually to an overshooting, i.e.,
\[
v_9 = \frac{4}{3}v_{8} - \frac{1}{3}v_{4} - 0.7 + 1 < -1,
\]
violating the lower bound of Theorem \ref{On the stability for constant inputs}. By placing additional sign changes at the extrema of the following trajectories, one can induce oscillations of increasing amplitude and, hence, instability of the system.
\begin{figure}[h!]
  \noindent\makebox[\linewidth][c]{%
    \resizebox{1.2\linewidth}{!}{%
      \begin{subfigure}{0.5\linewidth}
        \centering
        \includegraphics[width=\linewidth]{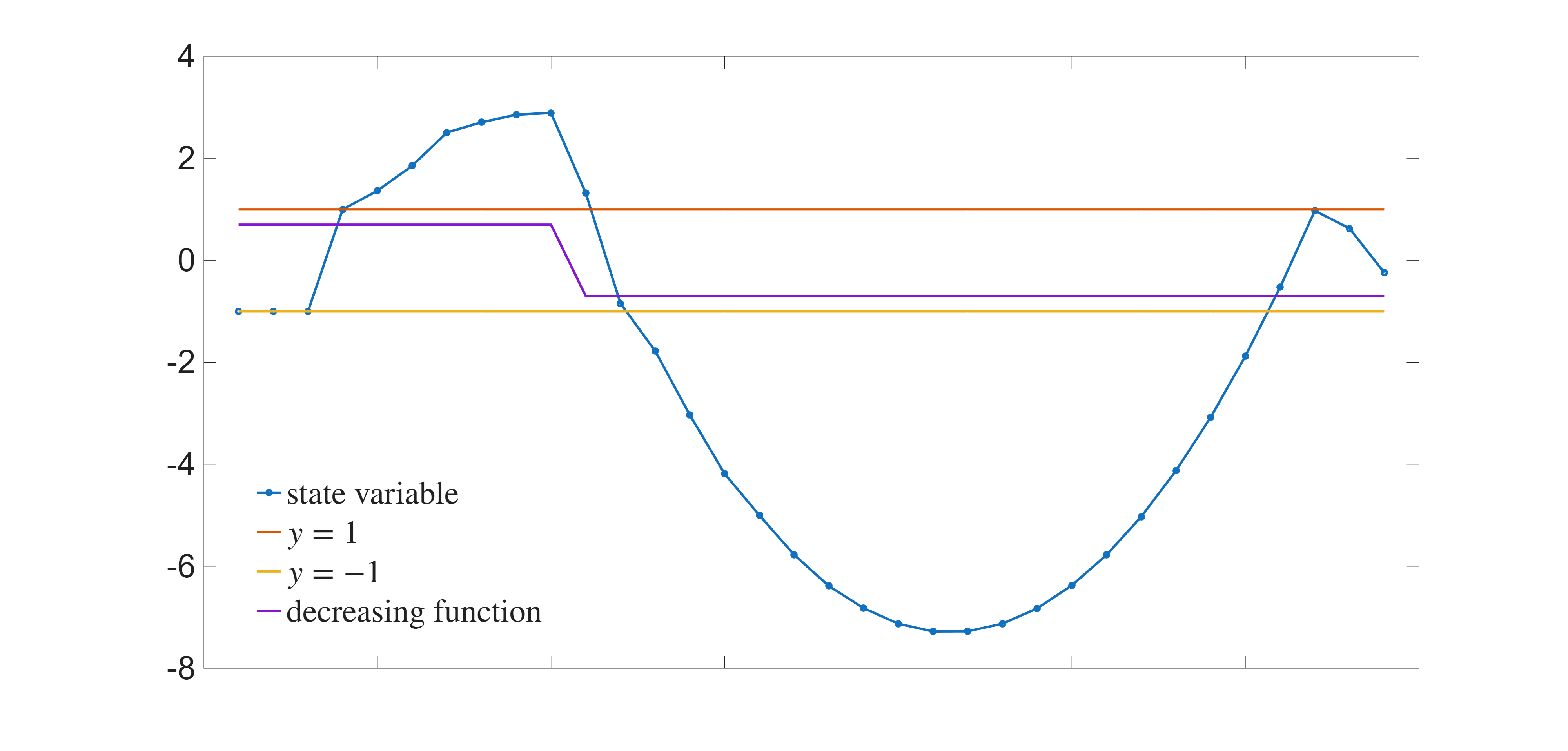}
        \caption{Step function with large magnitude of decrease}
        \label{fig:fig1}
      \end{subfigure}
      \begin{subfigure}{0.5\linewidth}
        \centering
        \includegraphics[width=\linewidth]{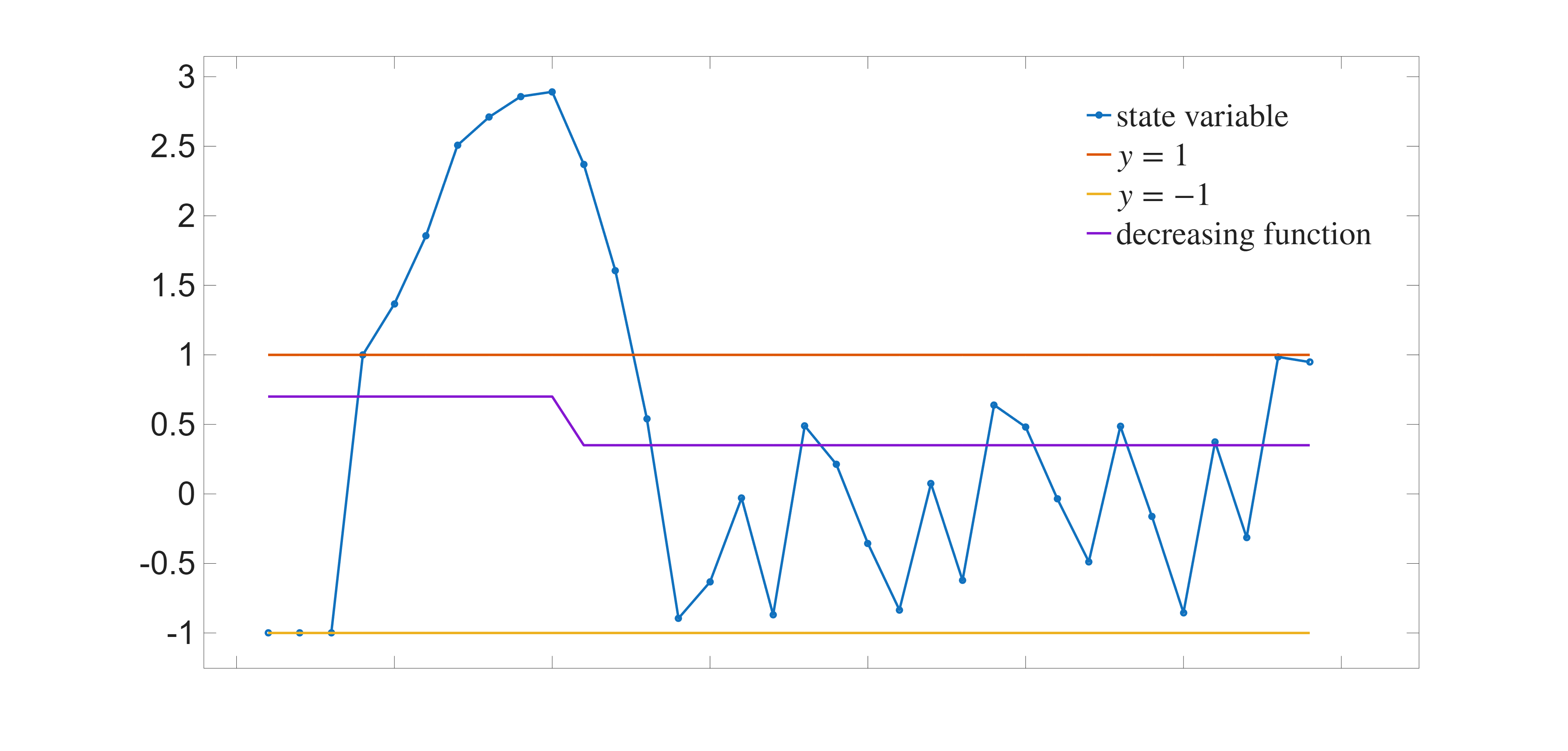}
        \caption{Step function with small magnitude of decrease}
        \label{fig:fig2}
      \end{subfigure}
    }%
  }
  \caption{{\small Plot of the state variables corresponding to two different sequences. In panel $(a)$, the sequence exhibits a sharp drop, decreasing rapidly from $0.7$ at step $n = 10$ to $-0.7$ at step $n = 11$. In contrast, in panel $(b)$ the function deacreases from $0.7$ at the step $n = 10$ to $0.35$ at the step $n = 11$. These examples highlight how the magnitude of variation between steps can affect the evolution of the state variables and thus the stability of the scheme.}}
  \label{fig:duefigure}
\end{figure}

% Assume that the step function from $n+6$ (and here $v_{n+6} = 2.89$) switches from $0.7$ to $-0.7$. As illustrated in Figure \ref{fig:fig1}, this change leads to an overshoot into $(-\infty,-1)$ after transitioning with one point through $[-1,1]$, since

% and thus $v_{n+9} \in (-\infty,-1)$.
% Therefore, the configuration $(v_{n+5},\dots, v_{n+8})$ triggers a trajectory in $(-\infty,-1).$ However, the peak of this trajectory will be lower than the one obtained by the negative critical trigger, as the values $v_{n+5},v_{n+6},v_{n+7}>1.$ It is intuitive to see that if this procedure is iterated, as discussed in the Preliminaries, it will give rise to oscillations of increasing magnitude.

However, for an input signal with a less extreme jump, e.g., $$y_n = \begin{cases}
    0.7 & \text{if} \ n\le 6\\
    0.35 & \text{if} \ n>7,
 \end{cases}$$ the effects of the acceleration can be controlled. 
 % if at step $6$ the function switches from $0.7$ to a constant value of $0.35$, we can observe from Figure \ref{fig:fig2} that the state variable remains in $[-1,1]$. 
 Indeed, one can observe that the maximum of the trajectory is $2.89$, which yields that after the trajectory, the state variable remains $\ge -1$ even when $y_n$ changes back and forth between the two values $0.7$ and $0.35$.\\
 As we will show in this section, a jump threshold below which overshooting can be avoided is given by the gap function, as introduced in the Definition \ref{Definition: gap function}.
 More precisely, if $y$ and $y'$ differ by less than the gap function, then the next trajectory must also be positive, as by definition of the gap function and the upper approximation $M(y,k)$, we have that for all $j$ between the two trajectories, either $q_j = 1,$ which implies that $v_j\ge-1$ by Lemma \ref{u_ningamma+}, or $q_j = -1$ and one also has
 \begin{equation}\label{Explanation gap function}
  v_{j} = \frac{4}{3}v_{j-1} - \frac{1}{3}v_{j-4} + y' -q_{j}\ge  -\frac43 -\frac{M(0.7,3)}{3} +\tilde{y}(0.7,3) +1 \ge -1.
\end{equation}
%  as all the previous values of the state variable are greater than $1$ and smaller than $2.89$, one has that
% \begin{equation}\label{Example_gap_function}
%     \frac{4}{3}v_i - \frac{1}{3}v_{i-4} + 0.35 + 1 > -1.
% \end{equation}

% which ensures that the following components of the state variable remain lower bounded by $-1$.

% We introduced this example because the gap function $g$ is designed specifically to quantify this threshold. More precisely, if given a trajectory produced by a constant input and said $M$ the peak value of the trajectory, we seek the smallest value $\tilde{y}(0.7)$ (and it corresponds to the threshold value $0.35$) such that, assuming to be in the worst case, i.e., $v_i = -1$ and $v_{i-4} = M,$ we still have that 
% \begin{equation}\label{Explanation gap function}
%     -\frac43 -\frac M3 +\tilde{y}(0.7) +1 \ge -1.
% \end{equation}
% In particular, in Definition \ref{Definition: gap function}, we have relaxed this condition by requiring that \eqref{Explanation gap function} holds not by considering the actual peak $M$ but for the upper approximation $M(0.7,3).$ 
% In this way, as long as we consider amplitudes $y$ such that $g(y)>0,$ we are allowed to have a controlled deacrese in the function value that still guarantees the boundness of the state variable.
\end{example}

Inspired by this reasoning, the key idea to achieve stability for the general case of bandlimited signals is the following:\\\\
\fbox{%
  \parbox{\dimexpr\linewidth-2\fboxsep-2\fboxrule\relax}{%
    Choose the sampling rate in such a way that in the course of a positive trajectory, the signal 
values change by at most the value of the gap function.
  }%
}\\

Then, with a similar reasoning, the lower bound of $-1$ still holds shortly after the end of the trajectory.

% This intuition is fundamental for proving the stability result in the general case of bandlimited functions.

% \textbf{Idea:} Based on Example~\eqref{example}, we can now understand that, when dealing with bandlimited functions, delicate cases as for the step functions no longer arise. Indeed, from  since bandlimited functions are Lipschitz functions with Lipschitz constant dependent from the bandwidth, by sufficiently oversampling, we can arbitrarily control the rate at which a function decreases over a given number of steps.  
% This observation, combined with the reasoning in Example~\eqref{example}, enables us to establish a minimum sampling rate that guarantees stability even when the maximum amplitude of the function and the filter do not satisfy the standard stability condition.

In what follows, we present a few preliminary properties that will be instrumental in proving the main result.

\begin{lemma}\label{majorant_state}
   Let $(y_n)_{n\in \mathbb{N}_0}$  be an input sequence such that $\|y\|_{\infty}<1$ and consider a $\Sigma\Delta$ scheme with feedback filter $h$ as in \ref{second_order_filter} for some $k\ge 3.$ Assume that at a step $N\in\mathbb{N}_0$ the configuration of previous state variables $\tilde{v} =(v_{N-(k+1)},\dots,v_{N-1})\in[-1,+\infty)^{k+1}$ is a generalized positive trigger with associated trajectory $(v_N,\dots,v_{N+S})\in (1,+\infty)^{S+1}$, with reach $S>0$. Then $y_N>1-\frac{2}{k},$ $S\le N(\bar{v},\|y\|_{\infty})$, with $\bar{v}$ being the critical trigger and
   \begin{align*}
       \max_{i =N,\dots,N+S}v_i
       \le \text{max}\left\{2h_1^n -1 + \left(\max_{t = N,\dots,N+S}y_{t}-1\right)\sum_{i = 0}^n h_1^{i-1} \ : \ n\in \mathbb{N}\right\}\\\le M\left(\max_{t = N,\dots,N+S}y_{t}\right).
   \end{align*}
   Moreover, it holds that $$\max_{t = N,\dots,N+S}y_n> 1-\frac{2}{k}$$
\end{lemma}

\begin{proof}
Without loss of generality, we assume that $N = 0.$ 
First, we show that $y_0>1-\frac{2}{k}$. If $y_{0}\le1-\frac2k$, then since the components of $(v_{-(k+1)},\dots, v_{-1})$ are $\ge -1$, we have that 
\begin{gather}\label{MLemma}
v_{N} = \frac{k+1}{k}v_{N-1}-\frac1k v_{N-(k+1)} + y_{m_1} - 1\\
\le \frac{k+1}{k}+\frac1k + 1-\frac{2}{k}-1 = 1,
\end{gather} 
where we used the fact that $-1\le v_{-1}\le 1$ as $v_{0}$ is the first element of the positive trajectory. Thus equation \eqref{MLemma} contradicts the fact that $v_0$ is the first element of the positive trajectory.\\
Proceeding analogously to Proposition \ref{evolution_highest_peak}, as  $q_n = 1$ for all $n = 0,\dots,S$, one obtains that for all $n = 0,\dots, S$
    \begin{gather*}
        v_n = h_{1}^nv_{-1} + \dots +h_{k+1}^{n}v_{-(k+1)} + \sum_{i = 0}^n (y_{n-i}-1)h_1^{i-1} \\
        \le h_{1}^nv_{-1} + \dots +h_{k+1}^{n}v_{-(k+1)} + \left(\max_{t = 0,\dots,S}y_n-1\right)\sum_{i = 0}^n h_1^{i-1}, 
    \end{gather*}
    where the inequality follows from the fact that $h_1^n\ge 0$ for all $n\in \mathbb{N}_0$, as established in Proposition \ref{lemma_evolution_h_i}.
    The inequality $M\le N(\bar{v},\|y\|_{\infty})$ follows from the monotonicity of the coverage in the second argument.
    The last inequality follows via the exact same argument for the constant case, see Proposition  \ref{parabola_behavior}.
\end{proof}

\begin{remark}\label{decreasing_step}
Following the intuition from Example \ref{example}, we aim to identify a sampling rate high enough to ensure that the samples are close enough to constant to guarantee that the $k+1$ state variables following a positive trajectory are $\ge -1.$ Analogously, the same sampling rate will ensure that the $k+1$  state variables following a negative trajectory will be $\le 1$. 

The key insight is that in the stabilizing input condition \eqref{stabilizing_equation}, the second and third terms correspond to different sampling points on or shortly after the trajectory, which will yield different sampling values in the non-constant case.  
Hence, the change of the function value needs to be small enough so that

\begin{equation}
        -\frac{k+1}{k} -    \frac{M(y_1,k)}{k} +y_2+1\ge -1,
\end{equation}   
not only for $y_1=y_2=y$, but also for arbitrary $y_1,y_2$ on or shortly after any trajectory.
A crucial property to tackle this problem is the strict monotonicity of the gap function, as it allows us to reduce this statement about many different trajectories to the single case of the trajectory generated by the critical trigger with the maximal amplitude as constant input. 

% More precisely, assume we aim to quantize a function $f$ with $h$ as in \eqref{second_order_filter} and $0\le y<y(k)$ be a stabilizing input and fix $d\in\mathbb{R}_+$ such that $y-d(N(\bar{v},y)+k)\ge\tilde{y}(y)$. If we consider $0\le \bar{y}<y,$ as the gap function is a strictly decreasing function, one has $\bar{y}-\tilde{y}(\bar{y})> y-\tilde{y}(y)\ge d(N(\bar{v},y)+k).$ In other words, if we use the decay rate of the maximal stable input, this will guarantee the stability for all the other smaller inputs. We will use this fact to prove the stability of the $\Sigma\Delta$ scheme for bandlimited signals. 
\end{remark}

\begin{theorem}[Stability for bandlimited signals]\label{Main_Theorem}
Let $h$ be a second order filter with $k\ge 3.$ Let $f\in PW_{1}(\mathbb{R})$, such that $\|f\|_{\infty}< 1- \frac{e^2}{(k+1)^2}$. Let 
  \[\lambda_0 = \frac{16e \|f\|_{\infty}}{k\left(1-\|f\|_{\infty}\right)\left(1-\frac{e^2}{(k+1)^2}-\|f\|_{\infty}\right)}
  \]

  If the samples $y_n = f(n/\lambda), \ {n\in \lambda I\cap \mathbb{N}_0},$  for an oversampling parameter of $\lambda \ge \lambda_0$, are quantized via the $\Sigma\Delta$ scheme with filter $h$ and with initial condition $\tilde{v} = (v_{-(k+1)},\dots,v_{-1})\in [-1,1]^{k+1},$ then the scheme is stable and
    \[
    \|v\|_{\infty}\le \begin{cases}
        1 & \text{if} \ \|f\|_{\infty}\le 1-2/k \\
         M(\|f\|_{\infty}) =  \|f\|_{\infty} + \frac{2}{k} +\frac{\left[  \mu-\gamma(1-\|f\|_{\infty})\right]^2}{2\alpha(1-\|f\|_{\infty})} & \text{if} \ \|f\|_{\infty}> 1-2/k
    \end{cases} 
    \]
    where $\gamma = \left(1+\frac{1}{k}\right)\left(1-\frac{1}{2k}\right),$ $\mu = \frac{2}{k}\left(1+\frac{1}{k}\right)^{k-1},$ and $\alpha = \frac{1}{k}\left(1+\frac{1}{k}\right).$
\end{theorem}

The proof intuition of this theorem is that for large enough $\lambda$, one can apply an estimate similar to \eqref{simp_max_traj} to show that all the trajectories have a length of at most 
\[ N:= \text{max}\left\{ n\in\ \mathbb{N} \ : \ 2h_1^n -1 + (\|f\|_{\infty}-1)\sum_{i = 0}^n h_1^{i-1} >1\right\}.\] 
Thus, for the stability bound we only have to take a maximum over a trajectory of at most that length.
This intuition motivates the following more technical version of the theorem above.
% \begin{theorem}\label{Main_Theorem}
% Let $h$ be a second order filter with $k\ge 3.$ Let $f\in \mathcal{C}^1(I),$ such that $1-\frac{2}{k}<\|f\|_{\infty}< 1- \frac{\left(1+\frac{1}{k}\right)^{2(k-1)}}{k^2}$ and $\|f'\|_{\infty} = 1.$ Let \[ N\ge \frac{4k}{k+1}\frac{\left(1+\frac{1}{k}\right)^{k-1}}{1-\|f\|_{\infty}} - k+1\] 
%   and 
%   \[\lambda \ge \frac{N}{\left(2-\frac{4}{k}-\frac{(e-1)^2}{k+1}\right)\left(1-\frac{\|f\|_{\infty}-\frac{k-2}{k}}{\frac{2}{k}-\frac{e^2}{k^2}}\right)}
%   \]
%   If $f$ is quantized in $I$ with an oversampling parameter of $\lambda$ and with initial condition $\bar{v} = (v_{-(k+1)},\dots,v_{-1})\in [-1,1]^{k+1},$ then the $\Sigma\Delta$ scheme is stable and \[
% \|v\|_{\infty}\le 
%     M(y) = y + \frac{2}{k} -\frac{\left[\beta(y-1) + \mu\right]^2}{2\alpha(y-1)},
%     \]
%     where $\beta = \left(1+\frac{1}{k}\right)\left(1-\frac{1}{2k}\right),$ $\mu = \frac{2}{k}\left(1+\frac{1}{k}\right)^{k-1},$ and $\alpha = \frac{1}{k}\left(1+\frac{1}{k}\right)$
% \end{theorem}

\begin{prop}\label{Main_Theorem2}
Let $h$ be a second order filter with $k\ge 3.$ Let $f\in \mathcal{C}^1(I)$ with $I = [0,T]$ and such that $1-\frac{2}{k}<\|f\|_{\infty}<y(k)<1$ and $\|f'\|_{\infty} \le L.$ Let \[ N(\bar{v},\|f\|_{\infty})= \text{max}\left\{ n\in\ \mathbb{N}_0 \ : \ 2h_1^n -1 + (\|f\|_{\infty}-1)\sum_{i = 0}^n h_1^{i-1} >1\right\}.\] 
If the samples $y_n = f(n/\lambda), \ {n\in \lambda I\cap \mathbb{N}_0},$  for an oversampling parameter of $\lambda \ge L\frac{N(\bar{v},\|f\|_{\infty})+k+1}{g(\|f\|_{\infty})}$, are quantized via the $\Sigma\Delta$ scheme with filter $h$ and with initial condition $\tilde{v} = (v_{-(k+1)},\dots,v_{-1})\in [-1,1]^{k+1},$ then the scheme is stable and \[
\|v\|_{\infty}\le \operatorname{max}_{n = 0,\dots,N}\left\{ \ 2h_1^n -1 + (\|f\|_{\infty} -1)\sum_{i= 0}^{n}h_1^{i-1} \right\}
.\]
\end{prop}
\begin{proof}
We first prove by contradiction that with the chosen oversampling, if the quantization scheme has initial value $\tilde{v}\in[-1,1]^{k+1},$ then for all $n\in \mathbb{N}_{0}$
\begin{equation}\label{division_property}
    (v_{n-(k+1)},\dots, v_{n})\in [-1,+\infty)^{k+2} \ \text{or} \ (v_{n-(k+1)},\dots, v_{n})\in (-\infty,1]^{k+2}. 
\end{equation}
Suppose that there exists an index such that this property does not hold, and consider the first instance $N\in\mathbb{N}$ at which this happens, i.e., some components of $(v_{N-(k+1)},\dots, v_{N})$ are $>1$ and some others are $<-1$. Let $t = \min \{ s\in \{N-(k+1),\dots,N\} \ : \ v_s>1 \}$ and  $t' = \min \{ s\in \{N-(k+1),\dots,N\} \ : \ v_s<-1 \}.$
By symmetry of the alphabet, we can assume without loss of generality that $t' = N$, and hence $v_s\ge-1$ for all $s = N-(k+2),\dots, N-1$.
As Property \eqref{division_property} is assumed to hold for all the indices previous to $N$ and  $\tilde{v}\in [-1,1]^{k+1}$, one must have $(v_{t-(k+1)},\dots,v_{t})\in[-1,+\infty)^{k+2}$.
We distinguish the two cases $t = N-(k+1)$ and $t>N-(k+1)$.\\
For both cases, we need the following observation. Since $\|f\|_{\infty} < y(k)$ and $g$ is a strictly decreasing function with $g(1-2/k)>0$, see Proposition \ref{Properties of the gap function}, one has $0<g(\|f\|_{\infty})= \|f\|_{\infty} -\tilde{y}(\|f\|_{\infty})$.
Thus for $\lambda \ge L\frac{N(\bar{v},\|f\|_{\infty})+k+1}{g(\|f\|_{\infty})}$ it holds that
\begin{equation}\label{choice_of_oversampling}
   \|f\|_{\infty} - L\frac{N(\bar{v},\|f\|_{\infty})+k}{\lambda}\ge \tilde{y}(\|f\|_{\infty}). 
\end{equation}

We also need the well-known fact that the variation of a Lipschitz function $f$ is controlled by its Lipschitz constant. Concretely, if a function $f$ is $L$-Lipschitz, then its samples $y_n = f\left(\frac{n}{\lambda}\right)$ satisfy for any $m_1\le m_2$ and $s_1,s_2\in\{m_1,\dots,m_2+k+1\}$
$$|y_{s_1}-y_{s_2}| = \left|f\left(\frac{s_1}{\lambda}\right)-f\left(\frac{s_2}{\lambda}\right)\right|\le L\left|\frac{s_1-s_2}{\lambda}\right|\le L\frac{m_2 +k+1 -m_1}{\lambda}.$$
As a consequence, 
\begin{equation}\label{min_maxof_yn}
    \min_{n = m_1,\dots,m_2+k+1}y_n\ge \max_{n=m_1,\dots,m_2+k+1}y_n - L\frac{m_2 - m_1 + k+1}{\lambda}.
\end{equation}
\textbf{Case 1:} $t = N-(k+1)$.\\
In this case $v_{N-(k+1)}>1$, so $v_{N-(k+1)}$ belongs to a positive trajectory  $(v_{m_1},\dots,v_{m_2})$ with 
\begin{equation}\label{bound_m_2_m_1}
m_1\le N-(k+1)\le m_2.    
\end{equation}
The triggering sequence $(v_{m_1 - (k+1)},\dots,v_{m_1 -1})$ generating this trajectory cannot have any elements $<-1$ as otherwise the sequence $(v_{m_1-(k+1)},\dots,v_{m_1})$ would violate Property \eqref{division_property}. 
Consequently, $(v_{m_1 - (k+1)},\dots,v_{m_1 -1})$ is a generalized positive triggering sequence and from Lemma \ref{majorant_state} we have that $m_2-m_1\le N(\bar{v},\|f\|_{\infty}),$ where $\bar{v}$ is the critical trigger.\\
To show that this leads to a contradiction, we need the following two key facts.\\
\textbf{Fact 1:}\\
\begin{equation}\label{y_n_bound_tildey}
    \max_{n = m_1,\dots,m_2+k+1} y_n - L\frac{N(\bar{v},\|f\|_{\infty})+k+1}{\lambda}\ge \tilde{y}\left(\max_{n = m_1,\dots,m_2+k+1} y_n \right)
\end{equation}
\textit{Proof of Fact 1:}\\
To show that \eqref{y_n_bound_tildey} holds, we first observe that since $(y_{m_1},\dots,y_{m_2})$ is a positive trajectory generated by a generalized positive trigger, then as shown in Lemma \ref{majorant_state} it must hold that $y_{m_1}>1-\frac{2}{k}$ and consequently $\max_{n = m_1,\dots,m_2+k+1}y_n>1-\frac{2}{k}$. Since $g$ is a strictly decreasing function in $\left[1-\frac{2}{k},1\right)$, it holds that 
\begin{gather*}
    \max_{n = m_1,\dots,m_2+k+1}y_n-\tilde{y}\left(\max_{n = m_1,\dots,m_2+k+1}y_n\right) = g\left(\max_{n = m_1,\dots,m_2+k+1}y_n\right)\\
    \ge g(\|y\|_{\infty})\ge L\frac{N(\bar{v},\|y\|_\infty)+k+1}{\lambda},
\end{gather*}
where the inequality follows from \eqref{choice_of_oversampling}.
\hfill \qed\\
\textbf{Fact 2:} For all $i = N-(k+1),\dots,N-1$ it holds that 
\begin{equation}
    -1\le v_i\le M\left(\max_{n = m_1,\dots,m_2+k+1} y_n\right).
\end{equation}
\textit{Proof of Fact $2$:}\\
By assumption $v_t = v_{n-(k+1)}>1,$ so the maximum $v_i$ belongs to some positive trajectory $(v_{m_3},\dots, v_{m_4})$. As by definition of the generalized positive triggering sequence, we have $v_{m_1-1}\le 1$ and by assumption $v_{N}<-1$, so $m_{1}\le m_3\le m_4 \le N-1\le m_2+k$, where the last inequality follows from \eqref{bound_m_2_m_1}. Then by Lemma \ref{majorant_state} one has $\max_{n = m_3,\dots,m_4}y_n > 1-\frac{2}{k}$ and for all $i = m_3,\dots,m_4$ $$-1\le v_i\le M\left(\max_{n = m_3,\dots,m_4}y_n\right)\le M\left(\max_{n = m_1,\dots,m_2+k+1}y_n\right),$$
where the last inequality comes from the fact that $$1-\frac{2}{k}<\max_{n = m_3,\dots,m_4}y_n \le \max_{n = m_1,\dots,m_2+k+1}y_n$$
and $M$ is increasing in $(1-2/k,1)$.
\hfill \qed\\
With Facts $ 1$ and $2$ we are now ready to establish the contradiction.

As $\|f'\|_{\infty}\le L,$ $f$ is $L$-Lipschitz and by \eqref{min_maxof_yn}, together with the fact that $N \le m_2 +k+1,$ we obtain that
\begin{equation}\label{bound_on_y_N}
    y_N \ge \max_{n=m_1,\dots,m_2+k+1}y_n - L\frac{m_2 - m_1 + k+1}{\lambda}\ge \max_{n=m_1,\dots,m_2+k+1}y_n - L\frac{N(\bar{v},\|f\|_{\infty}) + k+1}{\lambda},
\end{equation}
where the last inequality follows from Lemma \ref{majorant_state}.\\
Furthermore observe that our assumption $v_N<-1$ entails that $q_N =-1$ and hence

\begin{gather}
    v_N = \frac{k+1}{k}v_{N-1} - \frac{1}{k}v_{N-(k+1)} + y_N +1 \ge -\frac{k+1}{k} -\frac{1}{k}M\left(\max_{n = m_1,\dots,m_2+k+1}y_n\right)+y_N +1\\
    \ge -\frac{k+1}{k} -\frac{1}{k}M\left(\max_{n = m_1,\dots,m_2+k+1}y_n\right)+ \max_{n = m_1,\dots,m_2+k+1}y_n -L\frac{N(\bar{v},\|y\|_{\infty})+k+1}{\lambda}+1\\
    \ge -\frac{k+1}{k} -\frac{1}{k}M\left(\max_{n = m_1,\dots,m_2+k+1}y_n\right)+ \tilde{y}\left(\max_{n = m_1,\dots,m_2+k+1}y_n\right)+1 = -1,
\end{gather}
which yields the desired contradiction.
Here, the first inequality uses Fact $2$, the second inequality follows from \eqref{bound_on_y_N}, the third inequality is based on Fact $1$, and the last equality exploits Definition \ref{Definition: gap function}.

\textbf{Case 2:} $t>N-(k+1)$.
In this case $v_{N-(k+1)}\in [-1,1]$ and again by assumption $v_N<-1.$ Thus the trajectory containing the maximum state variable $v_i$ among $i = N-(k+1),\dots,N-1$ must lie fully in $\{N-(k+1),\dots,N-1\}$. Consequently, we can conclude analogously to Case $1$ that for all $ i = N-(k+1),\dots, N-1$

% As $v_N<-1$, $v_t>1$ is the first element of a positive trajectory $(v_t,\dots,v_{t_1})$ with $t\le t_1\le N-1$ and Property \eqref{division_property} holds for $n = t<N$, we have from Lemma \ref{majorant_state} that $y_t>1-\frac{2}{k}$ and for all $i = t,\dots,t_1$

$$-1\le v_i\le  M\left(\max_{n = N-(k+1),\dots,N}y_n\right).$$ 
Moreover, in direct analogy to Fact $1$ in Case $1$ one obtains
\begin{gather}
    y_N \ge  \tilde{y}\left(\max_{n = N-(k+1),\dots,N}y_n\right).
\end{gather}
Therefore, as $v_N<-1$ and thus $q_N = -1,$ we have
\begin{gather*}
    v_n = \frac{k+1}{k}v_{N-1} -\frac{1}{k} v_{N-(k+1)} + y_N +1 \\
    \ge -\frac{k+1}{k} -\frac{1}{k} + \tilde{y}\left(\max_{n = N-(k+1),\dots,N}y_n\right)+1\\
    \ge -\frac{k+1}{k} - \frac{M\left(\max_{n = N-(k+1),\dots,N}y_n\right)}{k} + \tilde{y}\left(\max_{n = N-(k+1),\dots,N}y_n\right)+1= -1,
\end{gather*}
which is again a contradiction.
Here, again, most steps are in analogy to Case $1$, but the first inequality uses, in addition, that $v_{N-(k+1)}\le1$, and the second inequality exploits that $M(z)>1$ for all $z >1-2/k$.\\
Thus, Property \eqref{division_property} holds for all $n\in \mathbb{N}_0$.\\
This entails a uniform bound on the state variable.
Indeed, by \eqref{division_property} it follows that if an element of the state trajectory is $>1$, then it must belong to a positive trajectory generated by a generalized positive trigger, and as a consequence of Lemma \ref{majorant_state}, it is bounded by 
$\max\left\{ n \ : \ 2h_1^n -1 +(\|f\|_{\infty}-1)\sum_{i = 0}^{n}h_1^{i-1}\right\}$.
This concludes the proof.
\end{proof}

From the above theorem, the extension to bandlimited functions follows directly. Indeed, by using the fact that bandlimited functions are Lipschitz functions, we can adjust the oversampling defined in the above theorem to guarantee stability in this case.
% \begin{coro}
%     Let $h$ be a second order filter with $k\ge 3.$ Let $f$ be a bandlimited function with bandwith $K$ and such that $1-\frac{2}{k}<\|f\|_{\infty}<y(k)<1.$ Let \[ N:= \text{max}\left\{ n\in\ \mathbb{N} \ : \ 2h_1^n -1 + (\|f\|_{\infty}-1)\sum_{i = 0}^n h_1^{i-1} >1\right\}\] 
%   and $\lambda \ge K \|f\|_{\infty}\frac{N+k+1}{g(\|f\|_{\infty})}.$\\ If $f$ is quantized with an oversampling parameter of $\lambda$ and with initial condition $\bar{v} = (v_{-(k+1)},\dots,v_{-1})\in [-1,1]^{k+1},$ then the $\Sigma\Delta$ scheme is stable and \[
% \|v\|_{\infty}\le \text{max}_{n = 0,\dots,N}\left\{ \ 2h_1^n -1 + (\|f\|_{\infty} -1)\sum_{i= 0}^{n}h_1^{i-1} \right\}
% .\]
% \end{coro}

\begin{proof}[Proof of Theorem \ref{Main_Theorem}]
    The case we discuss is $\|f\|_{\infty}>1-\frac{2}{k}$, as otherwise the stability criterion \eqref{classical_criterion} guarantees stability independently of the sampling rate. No additional oversampling is needed.
    From Bernstein's inequality (see \cite{nikol2012approximation}), we know that as $f$ is $1$-bandlimited \[
    \|f'\|_{\infty}\le \|f\|_{\infty}.
    \]
    Hence, since by $iii)$ of Proposition \ref{Properties of the gap function} it holds that $\|f\|_{\infty}<1-\frac{e^2}{(k+1)^2}<y(k)$, one can apply Proposition \ref{Main_Theorem2} with $L = \|f\|_{\infty}.$ 
    The oversampling parameter is obtained by direct bounds on the quantities $N(\bar{v},\|f\|_{\infty})$ and $g$.\\
    \textbf{Bound for $N(\bar{v},\|f\|_{\infty})$}: Let $p(k,n)$ denote the parabolic bound in Proposition \ref{parabola_behavior} for $y:= \|f\|_{\infty}$. Then by Proposition \ref{parabola_behavior} $p(k,n)$ is an upper bound for the trajectory generated by $\bar{v}$ with constant input $\|f\|_{\infty}$.  Consequently  $N(\bar{v},\|f\|_{\infty})$ is upper bounded by the larger zero of the parabola $q(n,k):= p(n,k) -1 = -an^2+bn+c$ where 
\[
a:=\frac{k+1}{2k^2}(1-\|f\|_{\infty}),
\qquad
b:=\frac{2}{k}\left(1+\frac{1}{k}\right)^{k-1}-\frac{(k+1)(2k-1)}{2k^2}(1-\|f\|_{\infty}),
\quad
c:=\frac2k-(1-\|f\|_{\infty}).
\]
It is straightforward to see that $a,b,c>0$ and via power series expansion of the square root it follows that the only positive root of $q$ is lower than $\frac{b}{a}+\frac{c}{b}.$ By observing that $c/b\le \frac{1}{\left(1+1/k\right)^{k-1}}\le 1$ and plugging in the values of $a,b$ in the first summand of the bound and $c$ into the given bound, we have that
\begin{equation}
    N(\bar{v},\|f\|_{\infty})\le \frac{4k}{k+1}\left(1+\frac1k\right)^{k-1}\frac1{1-\|f\|_{\infty}}
-(2k-2), 
\end{equation}
and thus using $k\ge 3$ we obtain 
\begin{equation}
    N(\bar{v},\|f\|_{\infty})+k+1 \le \frac{4e}{1-\|f\|_{\infty}} 
\end{equation}
\textbf{Lower bound for $g$:} From $i)$ of Proposition \ref{Properties of the gap function} we recall that the gap function is a concave and non-negative on $\left[1-\frac{2}{k},1-\frac{e^2}{(k+1)^2}\right]$. Therefore, we have that for all $1-\frac{2}{k} <\|f\|_{\infty}\le 1-\frac{e^2}{(k+1)^2}$
\begin{gather}
    g(\|f\|_{\infty})\ge 
\frac{1-\frac{e^2}{(k+1)^2}-\|f\|_{\infty}}{ \frac{2}{k} -\frac{e^2}{(k+1)^2}}\,g\left(1-\frac{2}{k}\right)
+
\frac{\|f\|_{\infty}-\left(1-\frac{2}{k}\right)}{ \frac{2}{k} -\frac{e^2}{(k+1)^2}}\,g\left(1-\frac{e^2}{(k+1)^2}\right)\\
\ge \frac{1-\frac{e^2}{(k+1)^2}-\|f\|_{\infty}}{ \frac{2}{k} -\frac{e^2}{(k+1)^2}}\,g\left(1-\frac{2}{k}\right)\ge \frac{k}{2}\left(1-\frac{e^2}{(k+1)^2}-\|f\|_{\infty}\right)\,g\left(1-\frac{2}{k}\right)\\
\ge \frac{k}{4}\left(1-\frac{e^2}{(k+1)^2}-\|f\|_{\infty}\right),
\end{gather}
where the last inequality follows again from $i)$ of Proposition \ref{Properties of the gap function}.
Combining these estimates yields the desired oversampling parameter $\lambda_0$.

\end{proof}

Theorem \ref{Main_Theorem} shows not only that the bound on the state variable depends on the maximal amplitude of the signal—and therefore grows as the signal amplitude increases—but also that its proof provides additional insight into where large values of the state variable may occur. For example, Property \eqref{division_property} ensures that whenever some components of the state variable exceed $1$, those components lie on a trajectory $(v_{m_1},\dots,v_{m_2})$ generated by a generalized positive triggering sequence and moreover, by Lemma \ref{majorant_state}
\[
\max_{i = m_1,\dots,m_2}v_i\le \max_{n}\left\{2h_1^n-1 +\left(\max_{i = m_1,\dots,m_2}y_i - 1\right) \sum_{i = 0}^{n}h_1^{i-1}\right\}.
\] 
If now, a second trajectory $(v_{m_3},\dots,v_{m_4})$ occurs later in the state sequence, when the function values are smaller in magnitude,  again by Lemma \ref{majorant_state}
\begin{gather*}
    \max_{i = m_3,\dots,m_4}v_i\le \max_{n}\left\{2h_1^n-1 +\left(\max_{i = m_3,\dots,m_4}y_i - 1\right) \sum_{i = 0}^{n}h_1^{i-1}\right\}
    % \le \max_{n}\left\{2h_1^n-1 +\left(\max_{i = m_1,\dots,m_2}y_i - 1\right) \sum_{i = 0}^{n}h_1^{i-1}\right\},
\end{gather*}
which is smaller than the bound for the first trajectory. Hence, the local bound on the state variable depends only on the magnitudes of the function samples during the trajectory's triggering sequence, not on the global maximum, and thus our theory explains what we observe in Figure \ref{fig:v_n_in_Gamma_+}. This is a significant improvement to the bounds of \cite{yilmaz2002stability}, where a large maximum admissible bound for the function can lead to large state variables even if the signal stays significantly below the bound, cf. Figure \ref{Yilmaz_plotz}.

\section{Conclusions and Future Work}
In this paper, we established the first stability guarantees for $\Sigma\Delta$ quantization schemes for general bandlimited inputs beyond first order that outperform the stability condition \eqref{classical_criterion}, which previously provided the best stability analysis for $\Sigma\Delta$ in the context of general bandlimited signals.
% In contrast to many previous approaches, our analysis describes the trajectories of the state variables rather than characterizing the invariant set, an approach that had previously been performed only in some specific example cases. The advantage of this viewpoint is that it allows for longer filters, which is challenging when analyzing invariant sets, due to the resulting large dimension, and perturbations resulting from non-constant inputs can be better incorporated. 
We demonstrated that for second-order $\Sigma\Delta$ schemes with sparse feedback filters as proposed by G\"unturk \cite{gunturk2003one}, our approach provides a qualitatively more accurate description of the dynamics of the state variable, which in turn leads to better bounds.

At the same time, we expect that our analysis technique will also yield improvements beyond second-order. However, a challenge for higher orders is that, even with constant inputs, the trajectories are no longer exactly parabolic. Thus, the higher-order perturbations need to be controlled in addition to the effects of the non-constant inputs. Similarly, we expect that our analysis technique can find additional applications in higher order $2$-dimensional $\Sigma\Delta,$ which is important for digital halftoning (see \cite{krahmer2023enhanced}). Here, working with amplitudes close to $1$ is particularly important, as restricting the amplitude corresponds to excluding dark colors. 

Lastly, we find it interesting to revisit the paradigms for filter design. Namely, the sparse filter design proposed in \cite{deift2011optimal} was directly attempting to minimize the $\ell^1$ norm reflecting the stability condition \eqref{classical_criterion}. Hence, it remains to be explored whether sparse filters are also preferred from the viewpoint of our refined analysis, or whether some dense filters that are worse in terms of condition \eqref{classical_criterion} can lead to improved trajectories. Specifically, we find it very interesting to analyze the trajectories of heuristic filter designs proposed by practitioners, such as the ones driven by Lee's criterion (see \cite{schreier2005understanding}), for which no rigorous stability analysis is available.

\bibliographystyle{unsrt}
\bibliography{references}
\end{document}